\documentclass[11pt]{article}
\usepackage[utf8]{inputenc}
\usepackage[T1]{fontenc}
\usepackage[margin=1in]{geometry}
\usepackage{natbib}
\setcitestyle{authoryear}
\bibpunct{[}{]}{;}{a}{}{,}
\let\cite\citep

\newif\ifanonymousmode
\anonymousmodefalse

\newif\ifshowendacks
\showendackstrue

\title{Simultaneous EF1 and approximate MMS allocations for submodular valuations}
\author{Uriel Feige\thanks{Weizmann Institute of Science, Rehovot, Israel. \texttt{uriel.feige@weizmann.ac.il}}
  \and
  Assaf Fine\thanks{Weizmann Institute of Science, Rehovot, Israel. \texttt{assaf.fine@weizmann.ac.il}}
}
\date{\today}

\newcommand{\renderbackmatter}{\renderbibliography
  \appendix
  \section{Missing proofs from Section~\ref{sec:RECE1}}\label{app:missing-ece}

{
\subsection{ECE outputs an EF1 allocation}\label{app:ECE_EF1}
}

{Here we prove Proposition~\ref{prp:ECE_EF1}.}

\begin{proof}
Since a new item $e$ is always allocated to an agent that no other agent envies, if $\hat{A}$ is EF1 just before $e$ is allocated, 
then after allocating $e$  to such an agent the updated partial allocation remains EF1, by definition.
It remains to show that EF1 is preserved by the Cycle Rotation Process. 
{This holds because cycle rotations only permute existing bundles among agents (no bundle's contents change), and each agent's value for her own bundle does not decrease.}

\end{proof}

\subsection{The allocation produced by RECE1 is EFL}\label{app:RECE1_EFL_as}

{Here we prove Proposition~\ref{prop:RECE1_EFL_as}.}

\begin{proof}[Proof sketch] 
Since RECE1 is an instantiation of ECE, Proposition~\ref{prp:ECE_EF1} gives EF1.
For the stronger EFL property, fix agents $i,j$ with $|A_j|\ge 2$ and let $g$ be the last item added to $A_j$ (necessarily in the random phase).
At the round $g$ was allocated, the recipient had no incoming edges in the envy graph, so removing $g$ eliminates $i$'s envy: $v_i(A_i) \ge v_i(A_j \setminus \{g\})$.
Moreover, $g$ was still unallocated during the greedy phase, so by the greedy choice $v_i(g) \le v_i(e_i) \le v_i(A_i)$.
Together, these two conditions establish EFL.
\end{proof}

\section{Missing proofs from Section~\ref{sec:W_i_expected}}\label{app:Expected_lower_bound}

\subsection{Proof of Claim \ref{clm:marginal-beta}} \label{app:marginal-beta}
This is a well known fact. For completeness, we present a short proof.
\begin{proof}
By monotonicity, $v(A\cup B)\ge v(B)$, and therefore
\[
v(B\mid A)=v(A\cup B)-v(A)\ \ge\ v(B)-v(A)\ \ge\ \beta-\alpha.
\]
Consider adding the items of $B=(e_1,\dots,e_{|B|})$ one by one to $A$.
Submodularity implies:
\[
v\left(B\mid A\right)
=\sum_{l=1}^{\left|B\right|}v\left(e_l\mid A\cup_{k<l}e_{k}\right)
\leq\sum_{l=1}^{\left|B\right|}v\left(e_l\mid A\right).
\]
Combining the two inequalities yields
\[
(\beta-\alpha)\ \le\ \sum_{l=1}^{|B|} v(e_l\mid A).
\]
Finally, the random item $e$ is uniform in $B$, so averaging over $|B|$ gives the claimed inequality.
\end{proof}

\subsection{Averaging bound for monotone submodular valuations}\label{app:expected_lower_sub}

{We prove the standard averaging bound for monotone submodular functions: if $S$ is a uniformly random $r$-subset of a set $B$ with $|B|=k$, then $\mathbb{E}[v(S)] \ge (r/k)\,v(B)$. }
{
\begin{proof}
By Lemma~\ref{lem:MMSoverAdditive}, there exists an additive $v' \le v$ with $v'(B) = v(B)$.  Since each $e \in B$ belongs to a uniformly random $r$-subset with probability $r/k$, linearity of expectation gives
\[
\mathbb{E}[v(S)] \;\ge\; \mathbb{E}[v'(S)] \;=\; \sum_{e \in B} v'(\{e\})\,\Pr[e \in S] \;=\; \frac{r}{k}\,v'(B) \;=\; \frac{r}{k}\,v(B). \qedhere
\]
\end{proof}}

\subsection{Calculation omitted from Proposition \ref{prop:expected_c_j}} \label{app:expected_c_j}

\begin{lemma}[Sum bound for $g(\alpha)$]\label{lem:g-sum-bound}
For every integer $k\ge 1$ and every $\alpha\in(0,1)$,
\[
\sum_{r=1}^{k}\Bigl(\frac{1}{k}-\frac{\alpha}{r}\Bigr)_+\ \ge\ g(\alpha).
\]
\end{lemma}

{
\begin{proof}
If $k\alpha<1$, then $\alpha/r\le\alpha<1/k$ for every $r\in\{1,\dots,k\}$, so all $k$ summands are positive and the LHS equals $1-\alpha H_k$, where $H_k:=\sum_{r=1}^{k}1/r$. Since $H_k\le 1+\ln k$ and $k<1/\alpha$,
\[
1-\alpha H_k\ \ge\ 1-\alpha(1+\ln k)\ \ge\ 1-\alpha-\alpha\ln(1/\alpha)\ =\ g(\alpha).
\]

Assume henceforth that $k\alpha\ge 1$. Let $a:=\lfloor \alpha k\rfloor\ge 1$, and write $x:=\alpha k$, so $x\in[a,a+1)$. The positive summands are exactly those with $r>a$, hence
\[
\sum_{r=1}^{k}\Bigl(\frac{1}{k}-\frac{\alpha}{r}\Bigr)_+\ =\ \frac{k-a}{k}-\alpha\sum_{r=a+1}^{k}\frac{1}{r}.
\]
Use the elementary bound
\[
\sum_{r=a+1}^{k}\frac{1}{r}\ \le\ \frac{1}{a+1}+\int_{a+1}^{k}\frac{dt}{t}\ =\ \frac{1}{a+1}+\ln\frac{k}{a+1}.
\]
Therefore
\[
\sum_{r=1}^{k}\Bigl(\frac{1}{k}-\frac{\alpha}{r}\Bigr)_+\ \ge\ 1-\frac{a}{k}-\frac{\alpha}{a+1}-\alpha\ln\frac{k}{a+1}.
\]
Subtracting $g(\alpha)=1-\alpha-\alpha\ln(1/\alpha)$ and using $1/\alpha=k/x$ (so $\ln(1/\alpha)=\ln(k/x)$), the remaining gap, divided by $\alpha$, equals
\[
\psi(x)\ :=\ 1-\frac{a}{x}-\frac{1}{a+1}-\ln\frac{x}{a+1}.
\]
We wish to show that \(\psi(x)\ge 0\) for every possible \(x\in[a,a+1)\).
Fix \(a\). Differentiating
with respect to \(x\) gives
\[
\psi'(x)\ =\ \frac{a}{x^2}-\frac{1}{x}\ =\ \frac{a-x}{x^2}\ \le\ 0,
\]
so $\psi$ is non-increasing on $[a,a+1)$, and
\[
\psi(x)\ \ge\ \lim_{y\uparrow a+1}\psi(y)\ =\ 1-\frac{a}{a+1}-\frac{1}{a+1}-\ln 1\ =\ 0.
\]
This proves the bound.
\end{proof}
}

{\subsection{Proof of the drop lemma (Lemma~\ref{lem:drop})}\label{app:drop-lemma}}

{
\begin{proof}
For every $j\in[n]$, every nonempty $R\subseteq B_j$, and every $e\in R$, we prove the two inequalities of the lemma:
\[
P(R)-v'_i(\{e\})\ \le\ P(R\setminus\{e\})\ \le\ P(R).
\]
Throughout, the function $(x-\alpha)_+=\max\{x-\alpha,0\}$ is monotone non-decreasing and $1$-Lipschitz on $\mathbb{R}$, so $0\le (x+a-\alpha)_+-(x-\alpha)_+\le a$ for every $x\ge 0$ and $a\ge 0$. We also use subset-additivity of $v'_i$ on $B_j$: $v'_i(T)=v'_i(T\setminus\{e\})+v'_i(\{e\})$ for every $T\subseteq B_j$ and every $e\in T$.

Fix $R$ and $e$ as above. Set $s:=|R|$, $a:=v'_i(\{e\})$, $R^-:=R\setminus\{e\}$, and abbreviate $D(R,e):=P(R)-P(R^-)$; the lemma's two inequalities are then equivalent to the single bound $D(R,e)\in[0,a]$, which we prove by induction on $s$.

\paragraph{Step 1: A recursion for $D(R,e)$.}
By Claim~\ref{obs:P-recursion} applied to $R$, splitting the inner expectation $\kaliEXPECTED_{e'\sim\mathbb{U}(R)}$ on whether $e'=e$ or $e'\in R^-$, and using $v'_i(R)=v'_i(R^-)+a$,
\begin{equation}\label{eq:drop-PR}
P(R)\ =\ \frac{(v'_i(R^-)+a-\alpha)_+}{s}\ +\ \frac{1}{s}\,P(R^-)\ +\ \frac{s-1}{s}\,\kaliEXPECTED_{e'\sim\mathbb{U}(R^-)}\bigl[P(R\setminus\{e'\})\bigr].
\end{equation}
For $s\ge 2$, applying Claim~\ref{obs:P-recursion} to the nonempty set $R^-$,
\begin{equation}\label{eq:drop-PRminus}
P(R^-)\ =\ \frac{(v'_i(R^-)-\alpha)_+}{s-1}\ +\ \kaliEXPECTED_{e'\sim\mathbb{U}(R^-)}\bigl[P(R^-\setminus\{e'\})\bigr].
\end{equation}
For $e'\in R^-$, $R\setminus\{e'\}\ni e$ and $R^-\setminus\{e'\}=(R\setminus\{e'\})\setminus\{e\}$, so $P(R\setminus\{e'\})-P(R^-\setminus\{e'\})=D(R\setminus\{e'\},e)$. From~\eqref{eq:drop-PR},
\[
D(R,e)\ =\ P(R)-P(R^-)\ =\ \frac{(v'_i(R^-)+a-\alpha)_+}{s}\ -\ \frac{s-1}{s}P(R^-)\ +\ \frac{s-1}{s}\,\kaliEXPECTED_{e'\sim\mathbb{U}(R^-)}\bigl[P(R\setminus\{e'\})\bigr].
\]
Substituting~\eqref{eq:drop-PRminus} for $P(R^-)$ and simplifying,
\begin{align*}
D(R,e)
&=\ \frac{(v'_i(R^-)+a-\alpha)_+}{s}\ -\ \frac{s-1}{s}\!\left[\frac{(v'_i(R^-)-\alpha)_+}{s-1}+\kaliEXPECTED_{e'}[P(R^-\setminus\{e'\})]\right]\\
&\quad +\ \frac{s-1}{s}\,\kaliEXPECTED_{e'}[P(R\setminus\{e'\})]\\
&=\ \frac{(v'_i(R^-)+a-\alpha)_+-(v'_i(R^-)-\alpha)_+}{s}\ +\ \frac{s-1}{s}\,\kaliEXPECTED_{e'\sim\mathbb{U}(R^-)}\bigl[D(R\setminus\{e'\},e)\bigr].
\end{align*}
Hence, for $s\ge 2$,
\begin{equation}\label{eq:drop-recur}
D(R,e)\ =\ \frac{(v'_i(R^-)+a-\alpha)_+-(v'_i(R^-)-\alpha)_+}{s}\ +\ \frac{s-1}{s}\,\kaliEXPECTED_{e'\sim\mathbb{U}(R^-)}\bigl[D(R\setminus\{e'\},e)\bigr].
\end{equation}

\paragraph{Step 2: Induction on $|R|$.}
\emph{Base ($s=1$, i.e., $R=\{e\}$).} By Claim~\ref{obs:P-recursion} applied to $\{e\}$ together with $P(\emptyset)=0$, $P(\{e\})=(a-\alpha)_+$. Hence $D(\{e\},e)=(a-\alpha)_+-0\in[0,a]$. 

\emph{Inductive step ($s\ge 2$).} Assume $D(R',e')\in[0,v'_i(\{e'\})]$ for every (subset, element) pair with subset of size $<s$. The first term of~\eqref{eq:drop-recur} is $(v'_i(R^-)+a-\alpha)_+-(v'_i(R^-)-\alpha)_+\in[0,a]$ (with $x:=v'_i(R^-)\ge 0$). The second term, $\kaliEXPECTED_{e'\sim\mathbb{U}(R^-)}[D(R\setminus\{e'\},e)]$, is an average of quantities each in $[0,a]$ by the inductive hypothesis (each $R\setminus\{e'\}$ contains $e$ and has size $s-1<s$). A convex combination, with weights $1/s$ and $(s-1)/s$, of values in $[0,a]$ lies in $[0,a]$, so $D(R,e)\in[0,a]$.

In either case $D(R,e)=P(R)-P(R^-)\in[0,a]$, equivalently $P(R)-v'_i(\{e\})\le P(R\setminus\{e\})\le P(R)$.
\end{proof}
}

\subsection{Proof of Proposition~\ref{prop:EW-true}}\label{app:EW-true}

\begin{proof}
Set $d_i(e):=v_i(e\mid A_e^{<e})$, so $W_i=\sum_{e\in\items}d_i(e)$. Fix an MMS bundle $B_j$, partition $B_j=G_j\sqcup R_j^0$ as in Proposition~\ref{prop:expected_c_j}, set $k:=|B_j|$ and $k':=|R_j^0|$.

\paragraph{Per-round bound.} The recipient bundle $A_{e_r}^{<e_r}$ is an envy-graph sink; combined with monotonicity of $v_i(A_i)$ along the run, $X_r:=v_i(A_{e_r}^{<e_r})\le v_i(A_i)$ deterministically, so $\kaliEXPECTED[X_r]\le\mu_i\le\alpha$. Applying Claim~\ref{clm:marginal-beta} to $v_i$ with $A:=A_{e_r}^{<e_r}$ and source set $R$ (a uniform $r$-subset of $R_j^0$ at round $r$), then the inequality $\kaliEXPECTED[\max\{Y,0\}]\ge \max\{\kaliEXPECTED[Y],0\}$ applied with $Y:=v_i(R)-\alpha$ and $\kaliEXPECTED[v_i(R)]\ge \kaliEXPECTED[v'_i(R)]=(r/k')\,v'_i(R_j^0)$ (Claim~\ref{clm:avg-submodular}, additivity of $v'_i$ on $B_j$, and $v'_i\le v_i$), gives
\begin{equation}\label{eq:EWtrue-perround}
\kaliEXPECTED[d_i(e_r)]\ \ge\ \Bigl(\tfrac{v'_i(R_j^0)}{k'}-\tfrac{\alpha}{r}\Bigr)_+.
\end{equation}

\paragraph{Reduction to the $W'_i$ argument.} From here the proof   {has a structure similar to that of the proof} of Proposition~\ref{prop:expected_c_j}, with $d_i$ in place of $c_i$, and~\eqref{eq:EWtrue-perround} in place of Claim~\ref{clm:per-round-contrib}. 

{Set $\tau:=v'_i(R_j^0)$ Since $v'_i$ is additive on $B_j$ with $v'_i(B_j)=1$ and $R_j^0\subseteq B_j$, $\tau\in[0,1]$. The greedy items of $B_j$ contribute
\[
\sum_{g\in G_j}d_i(g)\ =\ \sum_{g\in G_j}v_i(\{g\})\ \ge\ \sum_{g\in G_j}v'_i(\{g\})\ =\ 1-\tau,
\]
using $v'_i\le v_i$ and additivity of $v'_i$ on $B_j$. We split into two cases.

\textit{Case 1: $\tau\le\alpha$.} The random-phase contribution is non-negative, so the greedy contribution alone gives
\[
\kaliEXPECTED\Bigl[\sum_{e\in B_j}d_i(e)\Bigr]\ \ge\ 1-\tau\ \ge\ 1-\alpha\ \ge\ g(\alpha),
\]
where the last inequality, {which can be written as} {$1-\alpha-g(\alpha)=\alpha\ln(1/\alpha)\ge 0$}, {holds} for $\alpha\in(0,1)$.

\textit{Case 2: $\tau>\alpha$.} Set $\theta:=\alpha/\tau\in(0,1)$. Summing~\eqref{eq:EWtrue-perround} over $r=1,\dots,k'$ and applying Lemma~\ref{lem:g-sum-bound} with parameter $\theta$,
\[
\kaliEXPECTED\Bigl[\sum_{e\in R_j^0}d_i(e)\Bigr]\ \ge\ \tau\sum_{r=1}^{k'}\Bigl(\tfrac{1}{k'}-\tfrac{\theta}{r}\Bigr)_+\ \ge\ \tau\,g(\theta)\ =\ \tau\,g(\alpha/\tau).
\]
Adding the greedy contribution,
\[
\kaliEXPECTED\Bigl[\sum_{e\in B_j}d_i(e)\Bigr]\ \ge\ (1-\tau)+\tau\,g(\alpha/\tau)\ =\ {1-\alpha-\alpha\ln(\tau/\alpha)}\ \ge\ {1-\alpha-\alpha\ln(1/\alpha)}\ =\ g(\alpha),
\]
using $\tau\le 1$ in the penultimate inequality.

In either case, $\kaliEXPECTED[\sum_{e\in B_j}d_i(e)]\ge g(\alpha)$. Summing over $j\in[n]$ yields $\kaliEXPECTED[W_i]\ge n\cdot g(\alpha)$.}\qedhere
\end{proof}

\section{Missing proofs from Section~\ref{sec:high_prob}}\label{app:missing-high_prob}

{\subsection{Proof of Proposition~\ref{prop:Mk-submart}}\label{app:Mk-submart}

\textbf{Proposition~\ref{prop:Mk-submart}} \textbf{.} \emph{$\{\Lambda_k\}_{k=0}^{T}$ is a submartingale with respect to $\{\mathcal{F}_k\}$, i.e.\ $\kaliEXPECTED[\Lambda_k-\Lambda_{k-1}\mid\mathcal{F}_{k-1}]\ge 0$ for every $k\in[T]$.}

\begin{proof}
Fix $k\in[T]$ and let $j^\star\in[n]$ be the (random) index for which $e_k\in B_{j^\star}$. It suffices to show
\begin{equation}\label{eq:Mk-conditional-goal}
\kaliEXPECTED[\Lambda_k-\Lambda_{k-1}\mid\mathcal{F}_{k-1},\,j^\star=j]\ \ge\ 0
\end{equation}
for every $j$, since averaging over $j^\star$ then yields the unconditional submartingale property.

\paragraph{Increment decomposition.} Only $R_j^{k-1}$ shrinks at round $k$; the other remaining sets are unchanged. Hence the only term in $\Lambda_k$ that changes besides the credit $c_i(e_k)$ is the bundle potential of $B_j$, giving
\begin{equation}\label{eq:Mk-increment}
\Lambda_k-\Lambda_{k-1}\ =\ c_i(e_k)\ -\ \Delta P_k,
\qquad
\Delta P_k\ :=\ P(R_j^{k-1})-P(R_j^{k-1}\setminus\{e_k\}).
\end{equation}
Goal~\eqref{eq:Mk-conditional-goal} thus reduces to
\begin{equation}\label{eq:Mk-credit-vs-drop}
\kaliEXPECTED[c_i(e_k)\mid\mathcal{F}_{k-1},\,j^\star=j]\ \ge\ \kaliEXPECTED[\Delta P_k\mid\mathcal{F}_{k-1},\,j^\star=j].
\end{equation}

\paragraph{Drop side (RHS of~\eqref{eq:Mk-credit-vs-drop}).} Conditional on $\mathcal{F}_{k-1}$ and $\{j^\star=j\}$, the item $e_k$ is uniformly distributed on $R_j^{k-1}$. Therefore
\[
\kaliEXPECTED[\Delta P_k\mid\mathcal{F}_{k-1},\,j^\star=j]\ =\ \kaliEXPECTED_{e'\sim\mathbb{U}(R_j^{k-1})}\bigl[P(R_j^{k-1})-P(R_j^{k-1}\setminus\{e'\})\bigr]\ =\ {\frac{(v'_i(R_j^{k-1})-\alpha)_+}{|R_j^{k-1}|}},
\]
where the last equality is Claim~\ref{obs:P-recursion} (the recursion for $P$, rearranged).

\paragraph{Credit side (LHS of~\eqref{eq:Mk-credit-vs-drop}).} Let $A:=A_{e_k}^{<e_k}$ denote the contents of the recipient bundle just before $e_k$ is allocated.

\emph{Case 1: $v'_i(A)<\alpha$.} Then $c_i(e_k)=v'_i(e_k\mid A)$. Apply Claim~\ref{clm:marginal-beta} to $v'_i$ with the conditioning set $A$ (which satisfies $v'_i(A)<\alpha$) and source set $B:=R_j^{k-1}\subseteq B_j$, using $\beta:=v'_i(R_j^{k-1})$. Together with the uniformity of $e_k$ on $R_j^{k-1}$,
\[
\kaliEXPECTED[c_i(e_k)\mid\mathcal{F}_{k-1},\,j^\star=j]\ =\ \kaliEXPECTED_{e\sim\mathbb{U}(R_j^{k-1})}[v'_i(e\mid A)]\ \ge\ \frac{(v'_i(R_j^{k-1})-\alpha)_+}{|R_j^{k-1}|}.
\]

\emph{Case 2: $v'_i(A)\ge\alpha$.} Then $c_i(e_k)=v'_i(\{e_k\})$. By the uniformity of $e_k$ on $R_j^{k-1}$ and subset-additivity of $v'_i$ on $B_j$,
\[
\kaliEXPECTED[c_i(e_k)\mid\mathcal{F}_{k-1},\,j^\star=j]\ =\ \frac{1}{|R_j^{k-1}|}\sum_{e\in R_j^{k-1}} v'_i(\{e\})\ =\ \frac{v'_i(R_j^{k-1})}{|R_j^{k-1}|}\ {\ge\ \frac{(v'_i(R_j^{k-1})-\alpha)_+}{|R_j^{k-1}|}},
\]
{since $(x-\alpha)_+ \le x$ for every $x\ge 0$.}

In both cases, the LHS of~\eqref{eq:Mk-credit-vs-drop} is at least {$(v'_i(R_j^{k-1})-\alpha)_+/|R_j^{k-1}|$}, which equals the RHS. This proves~\eqref{eq:Mk-credit-vs-drop}, hence~\eqref{eq:Mk-conditional-goal}. \qedhere
\end{proof}

}

\subsection{Proof of Lemma \ref{lem:MMSoverAdditive}}\label{app:vprime}

{In our proof of Lemma \ref{lem:MMSoverAdditive} we make use of the following claim.}

\begin{claim}
\label{cl:decreaseSubmodular}
Let $v$ be a monotone submodular valuation function and consider any item $e\in\items$.
Then, for every $0 \le t {\le} v(\{e\})$, the following function $v'$ is a monotone submodular valuation function:
\begin{enumerate}[label=\textbf{\arabic*.}, leftmargin=2.2em]
    \item $v'(\{e\}) = t$.
    \item $v'(S) = v(S)$ for every $S \subseteq (\items \setminus \{e\})$.
    \item $v'(S \cup \{e\}) = \min\{\,v(S \cup \{e\}),\ v'(S) + t\,\}$ for every $S \subseteq (\items \setminus \{e\})$.
\end{enumerate}
Moreover, $v'(S)\le v(S)$ for every $S\subseteq \items$.
\end{claim}

\begin{proof}
\textbf{Pointwise domination.}
If $e\notin S$, then $v'(S)=v(S)$. If $e\in S$, write $S=T\cup\{e\}$ with $T\subseteq \items\setminus\{e\}$.
Then
\[
v'(S)=v'(T\cup\{e\})=\min\{v(T\cup\{e\}),\,v'(T)+t\}\le v(T\cup\{e\})=v(S).
\]

\medskip
\noindent\textbf{Monotonicity.}
Let $S\subseteq \items$ and $x\in \items\setminus S$. We prove $v'(S\cup\{x\})\ge v'(S)$.

If $e\notin S$ and $x\neq e$, then $v'(S)=v(S)$ and $v'(S\cup\{x\})=v(S\cup\{x\})\ge v(S)=v'(S)$.

If $x=e$ and $e\notin S$, then
\[
v'(S\cup\{e\})=\min\{v(S\cup\{e\}),\,v(S)+t\}\ge v(S),
\]
since both terms inside the minimum are at least $v'(S)$.

If $e\in S$ and $x\neq e$, write $S=T\cup\{e\}$ with $T\subseteq \items\setminus\{e\}$.
Then $S\cup\{x\}=(T\cup\{x\})\cup\{e\}$, and using monotonicity of $v$ and Item~2,
\[
v(T\cup\{x\}\cup\{e\})\ge v(T\cup\{e\}),\qquad v'(T\cup\{x\})=v(T\cup\{x\})\ge v(T)=v'(T).
\]
Therefore
\[
v'(S\cup\{x\})
=\min\{v(T\cup\{x\}\cup\{e\}),\,v'(T\cup\{x\})+t\}
\ge \min\{v(T\cup\{e\}),\,v'(T)+t\}=v'(S).
\]

\medskip
\noindent\textbf{Submodularity.}
We prove $v'(S)+v'(T)\ge v'(S\cup T)+v'(S\cap T)$ for all $S,T\subseteq\items$.

If $e\notin S\cup T$, then $v'=v$ on all four sets and the inequality holds since $v$ is submodular.

Assume $e\in S\cup T$. We handle the remaining cases.

\smallskip
\noindent\emph{Case 1: $e\in S$ and $e\notin T$.}
Write $S=S_0\cup\{e\}$ where $S_0\subseteq \items\setminus\{e\}$, and note $T\subseteq \items\setminus\{e\}$.
Then $v'(T)=v(T)$ and $v'(S\cap T)=v(S_0\cap T)$.
Also,
\[
v'(S)=\min\{v(S_0\cup\{e\}),\,v(S_0)+t\},\quad
v'(S\cup T)=\min\{v(S_0\cup T\cup\{e\}),\,v(S_0\cup T)+t\}.
\]
If $v'(S)=v(S_0\cup\{e\})$, then using $v'(S\cup T)=v'(S_0\cup T\cup\{e\})\le v(S_0\cup T\cup\{e\})$,
\[
v'(S)+v'(T)=v(S_0\cup\{e\})+v(T)\ge v(S_0\cup T\cup\{e\})+v(S_0\cap T)\ge v'(S\cup T)+v'(S\cap T),
\]
by submodularity of $v$.
If $v'(S)=v(S_0)+t$, then using $v'(S\cup T)\le v(S_0\cup T)+t$,
\[
v'(S)+v'(T)=v(S_0)+t+v(T)\ge v(S_0\cup T)+v(S_0\cap T)+t\ge v'(S\cup T)+v'(S\cap T),
\]
again by submodularity of $v$.

\smallskip
\noindent\emph{Case 2: $e\in S$ and $e\in T$.}
Write $S=S_0\cup\{e\}$ and $T=T_0\cup\{e\}$ where $S_0,T_0\subseteq \items\setminus\{e\}$.
By the definition of $v'$, we have
\[
v'(S)=\min\{\,v(S),\,v(S_0)+t\,\},\qquad
v'(T)=\min\{\,v(T),\,v(T_0)+t\,\}.
\]
Also,
\[
S\cup T=(S_0\cup T_0)\cup\{e\},\qquad S\cap T=(S_0\cap T_0)\cup\{e\},
\]
and therefore
\begin{align*}
v'(S\cup T)
&=\min\{\,v(S\cup T),\,v(S_0\cup T_0)+t\,\},\\
v'(S\cap T)
&=\min\{\,v(S\cap T),\,v(S_0\cap T_0)+t\,\}.
\end{align*}

We prove $v'(S)+v'(T)\ge v'(S\cup T)+v'(S\cap T)$ by a case analysis according to which term
attains each minimum.

\smallskip
\noindent\underline{Subcase 2.1:} $v'(S)=v(S)$ and $v'(T)=v(T)$.
Then $v'(S\cup T)\le v(S\cup T)$ and $v'(S\cap T)\le v(S\cap T)$, hence
\[
v'(S)+v'(T)=v(S)+v(T)\ge v(S\cup T)+v(S\cap T)\ge v'(S\cup T)+v'(S\cap T),
\]
where the first inequality is submodularity of $v$.

\smallskip
\noindent\underline{Subcase 2.2:} $v'(S)=v(S)$ and $v'(T)=v(T_0)+t$.
Since $T=T_0\cup\{e\}$, we also have $S\cup T=S\cup T_0$ and $S\cap T=(S\cap T_0)\cup\{e\}$.
Moreover, $v'(S\cup T)\le v(S\cup T)$ and $v'(S\cap T)\le v(S\cap T_0)+t$.
Therefore,
\begin{align*}
v'(S)+v'(T)
&=v(S)+v(T_0)+t\\
&\ge v(S\cup T_0)+v(S\cap T_0)+t \\
&= v(S\cup T)+\bigl(v(S\cap T_0)+t\bigr)\\
&\ge v'(S\cup T)+v'(S\cap T),
\end{align*}
where the inequality is submodularity of $v$ applied to the pair $(S,T_0)$.

\smallskip
\noindent\underline{Subcase 2.3:} $v'(S)=v(S_0)+t$ and $v'(T)=v(T)$.
This is symmetric to Subcase 2.2 (swap the roles of $S$ and $T$).

\smallskip
\noindent\underline{Subcase 2.4:} $v'(S)=v(S_0)+t$ and $v'(T)=v(T_0)+t$.
In this subcase,
\[
v'(S)+v'(T)=v(S_0)+v(T_0)+2t.
\]
Also, $v'(S\cup T)\le v(S_0\cup T_0)+t$ and $v'(S\cap T)\le v(S_0\cap T_0)+t$.
Thus,
\begin{align*}
v'(S\cup T)+v'(S\cap T)
&\le \bigl(v(S_0\cup T_0)+t\bigr)+\bigl(v(S_0\cap T_0)+t\bigr)\\
&= v(S_0\cup T_0)+v(S_0\cap T_0)+2t\\
&\le v(S_0)+v(T_0)+2t\\
&= v'(S)+v'(T),
\end{align*}
where the second inequality is submodularity of $v$ applied to $(S_0,T_0)$.

This completes the case analysis for $e\in S$ and $e\in T$, and therefore establishes submodularity of $v'$.

Thus $v'$ is submodular in all cases, completing the proof.
\end{proof}

\subsubsection{Additivizing disjoint bundles while preserving submodularity}

\paragraph{Lemma \ref{lem:MMSoverAdditive}}\label{par:MMSoverAdditive}
Let $v$ be a monotone submodular valuation function and let $B_1,\ldots,B_n$ be disjoint subsets of $\items$.
Then there exists a monotone submodular function $v'$ with the following properties:
\begin{itemize}
    \item $v'(S)\le v(S)$ for every $S\subseteq \items$.
    \item For each $j\in[n]$,
    \[
    v'(B_j)=v(B_j)
    \qquad\text{and}\qquad
    \sum_{e\in B_j} v'(\{e\}) = v'(B_j).
    \]
  {(Note: submodularity of $v'$ together with $\sum_{e\in B_j} v'(\{e\}) = v'(B_j)$ implies that $v'$ is additive on every subset of each $B_j$: for every $S\subseteq B_j$, $v'(S)=\sum_{e\in S} v'(\{e\})$.)}
\end{itemize}

\begin{proof}
Fix a bundle index $j$ and fix an \emph{internal order} of the items in $B_j$:
\[
B_j=\{e^j_1,\ldots,e^j_{k_j}\}.
\]
Define the chain prefixes $P^j_r:=\{e^j_1,\ldots,e^j_r\}$ (with $P^j_0=\emptyset$) and define additive weights
\[
w(e^j_r)\ :=\ v(P^j_r)-v(P^j_{r-1}) \ =\ v(e^j_r\mid P^j_{r-1}).
\]
By telescoping,
\[
\sum_{r=1}^{k_j} w(e^j_r) \ =\ v(P^j_{k_j})-v(P^j_0)\ =\ v(B_j).
\]

\medskip
\noindent\textbf{Step 1: $w$ lower-bounds $v$ on subsets of $B_j$.}

{By Claim~\ref{clm:chain-marginal-lower-bound}, applied with $B:=B_j$ under the internal order $e^j_1,\ldots,e^j_{k_j}$, the chain-order marginals $w(e^j_r)=v(e^j_r\mid P^j_{r-1})$ satisfy
\[
\sum_{e\in S} w(e)\ \le\ v(S)\qquad\text{for every }S\subseteq B_j,
\]
with equality on $S=B_j$ .}

\medskip
\noindent\textbf{Step 2: Build $v'$ by iterated truncations.}
Start with $v^{(0)}:=v$.
We will process items in a global sequence that respects the internal order inside each bundle:
process $e^1_1,e^1_2,\ldots,e^1_{k_1}$, then $e^2_1,\ldots,e^2_{k_2}$, and so on.

Suppose we are processing an item $e=e^j_r$.
Let the current function be $v^{(\ell)}$.
We apply Claim~\ref{cl:decreaseSubmodular} to $v^{(\ell)}$ and $e$ with parameter $t:=w(e)$, obtaining $v^{(\ell+1)}$.
By Claim~\ref{cl:decreaseSubmodular}, each step preserves {monotonicity}, submodularity and ensures pointwise domination:
\[
v^{(\ell+1)}(S)\le v^{(\ell)}(S)\le v(S)\qquad\forall S.
\]
Moreover, this truncation forces the \emph{maximum possible marginal} of $e$ into any set to be at most $t=w(e)$,
while leaving all sets not containing $e$ unchanged.

Because the items in different bundles are disjoint, truncating an item in $B_j$ never changes values of sets contained in $B_{j'}$
for $j'\neq j$ (those sets do not contain $e$). Hence we may analyze each bundle separately.

\medskip
\noindent\textbf{Step 3: {The final $v'$ satisfies $v'(B_j)=v(B_j)$ on every bundle $B_j$, and its restriction to each $B_j$ is additive with $v'(\{e\})=w(e)$.}}

Fix $j$ and consider the moment immediately after processing all items in $B_j$ in the iterated truncations.
We claim that for every $r\in\{0,1,\ldots,k_j\}$, the length-$r$ prefix $P^j_r=\{e^j_1,\ldots,e^j_r\}$ of $B_j$ (in the fixed internal order), at every iteration of the process after the truncation of item $e^j_r$, satisfies
\[
v'(P^j_r)\ =\ \sum_{u=1}^{r} w(e^j_u),
\]
and in particular $v'(B_j)=\sum_{u=1}^{k_j} w(e^j_u)=v(B_j)$.

We prove this by induction on $r$.
For $r=0$, $v'(\emptyset)=0$.
Assume it holds for $r-1$.
When processing $e^j_r$, we apply Claim~\ref{cl:decreaseSubmodular} with $t=w(e^j_r)$.
At that moment, the set $P^j_{r-1}$ contains no unprocessed items from $B_j$ and does not contain $e^j_r$. 
Therefore the new value of the prefix including $e^j_r$ is
\[
v'(P^j_r)=v'(P^j_{r-1}\cup\{e^j_r\})=\min\{\,v(\,P^j_r\,),\ v'(P^j_{r-1})+w(e^j_r)\,\}.
\]
By the induction hypothesis, $v'(P^j_{r-1})=\sum_{u=1}^{r-1} w(e^j_u)$, so the second term is
$\sum_{u=1}^{r} w(e^j_u)=v(P^j_r)$ by the definition of $w$ (telescoping along the chain).
Hence the two arguments of the minimum are equal, and we get
\[
v'(P^j_r)=v'(P^j_{r-1})+w(e^j_r)=\sum_{u=1}^{r} w(e^j_u).
\]
{By Claim \ref{cl:decreaseSubmodular}, this value remains the same for all the subsequent truncations.} This completes the induction and proves $v'(B_j)=v(B_j)$.

Finally, because the construction sets $v'(\{e^j_r\})=w(e^j_r)$ at the moment $e^j_r$ is processed, {and this value is preserved by every subsequent truncation: each later step processes some other item $e'\ne e^j_r$, and by Claim~\ref{cl:decreaseSubmodular} such a truncation leaves the value of every set not containing $e'$ unchanged. In particular, of $\{e^j_r\}$, so $v'(\{e^j_r\})=w(e^j_r)$ remains true at the end of the construction.} Therefore, we have
\[
\sum_{e\in B_j} v'(\{e\}) = \sum_{r=1}^{k_j} w(e^j_r) = v'(B_j).
\]
\end{proof}

}

\usepackage{booktabs} \usepackage[ruled]{algorithm2e} 
\SetAlFnt{\small}
\SetAlCapFnt{\small}
\SetAlCapNameFnt{\small}
\SetAlCapHSkip{0pt}
\IncMargin{-\parindent}

\usepackage{amsmath,amssymb,amsthm,mathtools}
\usepackage{pdfpages} \usepackage{hyperref}
\usepackage{enumitem}
\usepackage{array}

\newcommand{\kaliEXPECTED}{\mathbb{E}}

\newcommand{\Prb}{\mathbb{P}}

\newcommand{\items}{\mathcal{M}}
\newcommand{\agents}{\mathcal{N}}
\newcommand{\bundles}{\mathcal{B}}
\newcommand{\allocs}{\mathcal{A}} 
\theoremstyle{plain}
\newtheorem{theorem}{Theorem}[section]
\newtheorem{lemma}[theorem]{Lemma}
\newtheorem{proposition}[theorem]{Proposition}
\newtheorem{claim}[theorem]{Claim}
\newtheorem{observation}[theorem]{Observation}
\newtheorem{definition}[theorem]{Definition}
\newtheorem{example}[theorem]{Example}

\newtheorem{corollary}[theorem]{Corollary}
\newtheorem{assumption}[theorem]{Assumption}

\theoremstyle{remark}
\newtheorem{remark}[theorem]{Remark}

\begin{document}

\maketitle
  \begin{abstract}There are two common classes of fairness notions that are considered when allocating $m$ indivisible items to $n$ agents of equal entitlements. One is that of share-based fairness notions, with the maximin share (MMS) and its relaxations to $\rho$-MMS being prominent representatives of this class. The other is that of comparison-based fairness notions, with envy-freeness (EF) and its relaxations such as EF1 being prominent representatives of this class. In general, no class offers good guarantees for the other class. In this work, we design allocations that simultaneously satisfy notions from both classes, and specifically, are $\rho$-MMS for constant $\rho$ and EF1 (in fact, also EFL). Such results were previously known when agents have additive valuations, and we prove such results for the more general class of submodular valuations.
\end{abstract}
  \newpage

\section{Introduction} \label{sec:introduction}

We consider the well-studied problem of fair allocation of a set $\items$ of $m$ indivisible items to $n$ agents of equal entitlement. Each agent $i$ has her own valuation function $v_i$ over subsets of items, where we assume that valuations are normalized ($v_i(\emptyset) = 0$) and monotone ($v(S) \le v(T)$ for all $S \subset T$). An allocation $A = A_1, \ldots, A_n$ is a partition of $\items$ into $n$ sets, where for every $i$, agent $i$ gets set $A_i$. Given the valuations of the agents, we desire to output an allocation that satisfies some fairness criteria.

There are two common principles {by which fairness} of allocations is judged. One is {\em share based}, and the other is {\em comparison based}. In both cases, each agent can individually judge whether an allocation is fair for her, and if all agents find the allocation fair, then we say that the allocation is acceptable. 

In share-based notions, an agent $i$ fixes a share value $s_i$ based on her valuation function (and the number $n$ of agents), and judges an allocation to be fair if the value of the bundle $A_i$ that she receives satisfies $v_i(A_i) \ge s_i$. The share notion that we consider in this paper is the maximin share (MMS), and its relaxations.

\begin{definition}
\label{def:mms}
For an agent $i$, the \emph{maximin share} (MMS) is

\[
MMS(\items, v_i, n) :=\ \max_{\text{partitions }(B_1,\dots,B_n)\text{ of }\items}\ \min_{j\in[n]} v_i(B_j).
\]

A partition $(B_1,\dots,B_n)$ attaining the maximum is called an \emph{MMS partition} of agent $i$.
In addition, for $\rho \in (0,1)$, we use $\rho$-MMS to denote the share whose value is a $\rho$ fraction of the value of the MMS.  
\end{definition}

In comparison-based notions, an agent evaluates the fairness of an allocation by comparing the bundle that she receives with the bundles received by other agents. In this paper we consider envy-freeness and some of its relaxations.

\begin{definition}[Envy‐Free (EF)]
Allocation \(A\) is EF if for every \(i,j\):
\[
v_i(A_i)\;\ge\;v_i(A_j).
\]
Allocation \(A\) is EF1 if for every \(i,j\), either $A_j=\emptyset$, or there exists an item \(g\in A_j\) such that
\[
v_i(A_i)\;\ge\;v_i(A_j\setminus\{g\}).
\]

Allocation \(A\) is EFL if for every \(i,j\) one of the following holds:
\begin{itemize}
    \item $|A_j| \leq 1$   
    \item There exists an item \(g\in A_j\) such that
        \[
            v_i(A_i)\;\ge\;v_i(A_j\setminus\{g\}).
        \]
        and in addition $v_i(g)\leq v_i(A_i)$.
\end{itemize}
Note that every allocation that is EFL is also EF1.
\end{definition}

Both share-based and comparison-based notions are well accepted fairness notions, and neither one of them implies the other (except in special cases). Our goal in this paper is to prove the existence of allocations that simultaneously satisfy both notions. 
Our main result is the following.

\begin{theorem}
    \label{thm:EFLMMSsubmodular}
    There is some universal constant $\rho > 0$ such that every allocation instance in which valuation functions are submodular has an allocation that is simultaneously EFL and $\rho$-MMS.
\end{theorem}

{In our proof of Theorem~\ref{thm:EFLMMSsubmodular}, $\rho \ge 0.142$.} Our proof of Theorem~\ref{thm:EFLMMSsubmodular} also provides a randomized algorithm that with high probability produces an allocation as guaranteed by the theorem {(with a slightly worse value of $\rho$)}. The algorithm runs in polynomial time, and needs only {comparison} query access to the valuations.
{(A query specifies two disjoint sets of items, and the agent provides a one-bit response, specifying whether she values the first set at least as much as the second.)}

\subsection{Some related work}

There is a hierarchy of valuation classes introduced in~\cite{journals/geb/LehmannLN06}. Prominent classes in this hierarchy in increasing level of generality are additive, submodular, XOS and subadditive.
{(Here and elsewhere we use the following notation. For an item $e\in \items$ and valuation $v$, we write $v(e)$ instead of $v(\{e\})$. For a set $S\subseteq \items$, the \emph{marginal value} of $e$ given $S$ is $v(e\mid S)\ :=\ v(S\cup\{e\})-v(S)$.)}

\begin{definition}[Valuation classes]\label{def:valuation-classes}
Let $v:2^{\items}\to \mathbb{R}_{\ge 0}$ be a (normalized) monotone valuation.
\begin{enumerate}
    \item \textbf{Additive.} $v$ is \emph{additive} if for every $S\subseteq \items$,
    \[
        v(S)=\sum_{e\in S} v(e).
    \]

    \item \textbf{Submodular.} $v$ is \emph{submodular} if it satisfies diminishing marginal {values}: for all $S\subseteq T\subseteq \items$ and all $e\in \items\setminus T$,
    \[
        {v(e\mid S)\ \ge\ v(e\mid T).}
    \]

    \item \textbf{XOS (a.k.a. fractionally subadditive).} $v$ is \emph{XOS} if there exist additive valuations
    $\{v^r\}_{r\in R}$ such that for every $S\subseteq \items$,
    \[
        v(S)=\max_{r\in R} v^r(S)
    \]

    \item \textbf{Subadditive.} $v$ is \emph{subadditive} if for all $S,T\subseteq \items$,
    \[
        v(S\cup T)\ \le\ v(S)+v(T).
    \]
\end{enumerate}
\end{definition}

It {turns out} that even for additive valuations, MMS allocations need not exist{~\cite{journals/jacm/KurokawaPW18}}. This motivates considering $\rho$-MMS allocations. Known lower and upper bounds on the values of $\rho$ for the four classes of interest are presented in Table~\ref{tab:MMS}.

\begin{table}[hbt!]
        \centering
\renewcommand{\arraystretch}{1.4}\begin{tabular}{|l||l|l|}
 \hline
 \multicolumn{3}{|c|}{Feasible $\rho$ for $\rho$-MMS} \\
 \hline
  & lower bound & upper bound  \\
 \hline\hline
 Additive valuations    & $\frac{7}{9}$ \; \; \cite{journals/corr/abs-2511-13056}                  & $\frac{39}{40}$ \; \; \cite{conf/wine/FeigeST21} \\
 \hline
 Submodular valuations  & $\frac{10}{27}$ \; \; \cite{journals/corr/abs-2303-12444}                & $\frac{2}{3}$ \; \; \cite{conf/ijcai/0001CMS25} \\
 \hline
 XOS valuations         & $\frac{4}{17}$ \; \; \cite{conf/sigecom/FeigeG25}                        & $\frac{1}{2}$ \; \; \cite{journals/ai/GhodsiHSSY22} \\
 \hline
 Subadditive valuations & $\Omega\!\left(\frac{1}{\log\log n}\right)$ \; \; \cite{journals/corr/abs-2506-21493} & $\frac{1}{2}$ \; \; \cite{journals/ai/GhodsiHSSY22} \\
 \hline
\end{tabular}
        \caption{Known lower and upper bounds on $\rho$ for the existence of $\rho$-MMS allocations. Additional results are surveyed in Section~\ref{sec:related}.}
        \label{tab:MMS}
\end{table}

EF allocations clearly need not exist (e.g., when the number of items is smaller than the number of agents). However, EF1 allocations exist for all monotone valuations. This was shown in~\cite{conf/sigecom/LiptonMMS04} by an algorithm that we refer to as {\em envy cycle elimination} (ECE). This algorithm proceeds in rounds, where in every round, one agent gets an item. The agent that gets an item is one that no other agent envies, and the existence of such an agent can be enforced by a procedure referred to as {\em cycle rotation}. See more details in Section~\ref{sec:ECE}.

The ECE algorithm {comes in} several versions that differ in how the next item to be allocated is determined. In all of them, the EF1 guarantee holds, but they differ with respect to other guarantees that they may offer. In the worst case, the item is determined by an adversary, and we refer to this version as {WECE} {(W for worst)}. A more useful version is one in which for every agent, the first item that she receives is the one of highest value for her among the remaining items. We refer to this as a {\em greedy} first choice, and use a suffix of~1 to denote the use of this version. For example, WECE1 denotes the version {in which each agent} makes a greedy first choice, but other than that, the order of items is adversarial. WECE1 has the advantage over WECE in that it produces allocations that are not only EF1, but also EFL.

For additive valuations, the difference between EF1 and EFL is significant if one wished to also achieve an MMS approximation. An EF1 allocation might only be $\frac{1}{n}$-MMS, as this well known example illustrates.

\begin{example}
    \label{ex:1}
    Suppose that all agents have the same additive valuation, with $n-1$ large items, each of value $n$, and $n$ small items, each of value~1. Then, the MMS is $n$ ($n-1$ bundles each have one large item, and all small items are in the remaining bundle), whereas there is an EF1 allocation in which one bundle has value of only 1 ($n-1$ bundles each contain one small and one large item, and the remaining bundle contains one small item). Moreover, WECE will produce this EF1 allocation (under the order in which all small items precede all large ones).
\end{example}

For additive valuations, EFL implies $\frac{n}{2n-1}$-MMS {($\frac{1}{2}$-MMS is claimed in~\cite{conf/aaai/BarmanBMN18}, and the proof there shows the stronger implication of $\frac{n}{2n-1}$-MMS)}.
We may consider also a version of ECE, that we refer to as GECE (G for greedy) in which in every {round (and not just in the first round)} agents are greedy. A related version is known to offer a somewhat better guarantee of {$\frac{4}{7}$-MMS {\cite{journals/tcs/AmanatidisMN20}.}} Moreover, there are other allocation procedures {\cite{conf/aaai/AkramiR25a, journals/corr/abs-2505-19961}} that produce allocations that are EFL and $\frac{2}{3}$-MMS. 

However, beyond additive valuations, we are not aware of any significant positive results. This is true already for submodular valuations. It is no longer true that EFL allocations imply a constant approximation to MMS. In fact, neither do EF allocations,

\begin{example}
    \label{ex:2}
    There are $n^2$ items arranged in $n$ rows, each with $n$ items. The first item in the first row is referred to as $e_1$. For $v_1$, the value of a set $S$ equals the number of rows that $S$ intersects, plus a bonus of $\epsilon$ if $e_1 \in S$. This $v_1$ is submodular and satisfies $MMS(\items, v_1, n) = n$. For each $i \ge 2$, $v_i$ is additive, with each item in row $i$ having value~1, $e_1$ having value $n {- \epsilon}$, and other items have negligible value. An allocation in which each agent $i$ gets all items in row $i$ is envy free, but agent~1 gets a value of only $1 + \epsilon$.
\end{example}

GECE can be adapted to submodular valuations, by giving the agent the item that adds the highest marginal to the bundle that she holds at the time. However, this adaptation fails to give better than $\frac{1 + \epsilon}{n}$-MMS in the above example, if the first agent to receive an item is agent~1 (who will pick $e_1$). 
Also, other known approaches that give EF1 allocations do not give constant approximations for the MMS when valuations are submodular. For example, for two share notions, MXS and RMMS, it is known that they can be achieved simultaneously with EF1 (and also EFL), see~\cite{conf/aaai/AkramiR25, journals/corr/abs-2505-19961}, but their value might be as small as $\frac{1}{n}$-MMS when valuations are submodular. We remark that also allocations that maximize Nash Social Welfare might not give a constant approximation to the MMS, not even when valuations are additive, a case in which they are known to imply EF1~\cite{journals/teco/CaragiannisKMPS19}.

The strongest previous result that we are aware of (in terms of combining EF1 and MMS approximation for valuations beyond additive) is simultaneous EFL and $\frac{1}{n}$-MMS, which holds for subadditive valuations. In fact, this is an immediate consequence of the existence of EFL allocations, and the fact that every EF1 allocation is at least $\frac{1}{n}$-MMS (for subadditive valuations). 

\begin{proposition} \label{prop:EF1_is1_over_n}
    \label{pro:EF1subadditive}
    For subadditive valuations, every EF1 allocation offers every agent at least $\frac{1}{n}$-MMS.
\end{proposition}

{As we did not find a reference to this proposition, we present its simple proof for completeness.}

\begin{proof}
    Consider agent $i$ in an EF1 allocation $A = A_1, \ldots, A_n$. Then for every bundle $A_j$ (for $j \not= i$) there is an item $e_j \in A_j$ such that $v_i(A_i) \ge v_i(A_j \setminus \{e_j\})$. In the MMS partition for agent $i$ there must be a bundle not containing any of these items $e_j$ (as there are only $n-1$ such items and $n$ parts), and even if this bundle contains all remaining items, the MMS value  is (by subadditivity) at most $v_i(A_i) + \sum_{j \not= i} v_i(A_j \setminus \{e_j\}) \le n\cdot v_i(A_i)$.
\end{proof}

\subsection{Overview of our proof}

Our proof of Theorem~\ref{thm:EFLMMSsubmodular} is based on an approach that is natural to try, but seems difficult to implement. The idea is to use the ECE algorithm, that guarantees an EF1 allocation, but choose the order of items in such a way that in addition it will provide a $\rho$-MMS guarantee. We refer to the choice of order that offers the best MMS guarantee as the best choice, and hence refer to this version of ECE as BECE (B for best). 

Of course, we do not insist on finding the absolute best order for the items, but rather are satisfied with finding a ``good" order, as long as this order  provides an $\Omega(1)$-MMS guarantee. But even coming up with a good order, let alone analyzing the MMS approximation that it offers, appears elusive.

On the one hand, the order must depend on the valuation functions of the agents. The other alternative, of having a so called oblivious order over the items, provably fails, even if this order is randomized. For example, a uniformly random order over the items {(a version that we refer to as RECE) will} with high probability produce allocations that are only $O(\frac{\log n}{n})$-MMS, on the instance of Example~\ref{ex:1}.

On the other hand, natural approaches to make the order depend on the valuations (such as the use of GECE) also fail, as exemplified by Example~\ref{ex:2} and variations of it.

Our approach is to use a version of ECE that we refer to as RECE1. Recall that the suffix of~1 indicates that the first item that each agent receives is greedy, and the prefix of R indicates that all other items are allocated in a random order (in each round, to the agent dictated by the ECE algorithm). Observe that neither Example~\ref{ex:1} nor Example~\ref{ex:2} serve as a negative example for RECE1. The first greedy step handles Example~\ref{ex:1}, whereas for Example~\ref{ex:2} to be a negative example, typical agents should get almost all their items from the same row, a highly unlikely event when the order of the items is random.

How do we lower bound $\rho$ for the $\rho$-MMS guarantee of RECE1? 
As RECE1 is a randomized algorithm (there is a random order over the items), we should settle for a lower bound that holds with positive probability. This suffices in order to imply the same $\rho$ for BECE with certainty. Our proof will in fact prove lower bounds that hold not just with positive probability, but also with high probability, implying that if RECE1 is used in practice, it is likely to produce a $\rho$-MMS allocation (and surely produces an EFL allocation). 

At a high level, our proof uses a standard methodology in probabilistic analysis. It lower bounds for every $i$ the expected fraction of the MMS that agent $i$ receives, it shows strong concentration results around the expectation, and then uses a union bound to infer that with positive probability, all agents simultaneously receive value not much lower than the expectation. However, implementing this methodology in our case raises some challenges. Below, we state the two main lemmas that we prove, discuss some challenges in proving them, and then explain the key ideas that allow us to overcome these challenges. 

The first lemma provides a lower bound on the expected value that an agent receives. 

\begin{lemma} \label{lemma:exp_gurantee}
    \label{lem:expectation}
    {Choose {$\alpha = 0.2219$}}.  For every agent holding a submodular valuation $i$, the expected value received by agent $i$ in the RECE1 algorithm is at least $\alpha \cdot MMS(\items, v_i, n)$. (Note that we do not require other agents to have submodular valuations.)
\end{lemma}

\begin{remark}
{Our proof of Lemma~\ref{lem:expectation} requires  ordering the agents at random in the greedy phase of RECE1. All other proofs in this paper, including that of Theorem~\ref{thm:EFLMMSsubmodular}, hold regardless of how agents are ordered.}
\end{remark}

The second lemma offers a concentration result, showing that it is highly unlikely that the agent will receive a value significantly lower than the expectation.

\begin{lemma}
    \label{lem:deviation}
    {For $\alpha = 0.2219$ }and every $0 < \delta < 1$, for every agent holding a submodular valuation $i$, the value received by agent $i$ in the RECE1 algorithm is at least $(1 - \delta)\alpha \cdot MMS(\items, v_i, n)$ with probability at least {$1 - e^{-\Omega(\delta^2 n_i)}$}. In this last expression, $n_i$ denotes the number of agents whose first item $e$ that they receive in RECE1 satisfies $v_i(e) \le \alpha \cdot MMS(\items, v_i, n)$.

\end{lemma}

Lemma~\ref{lem:deviation} (in combination with a union bound on the $n$ agents) implies that as $n$ grows, {if all $n_i$ grow at a sufficiently large rate as a function of $n$,} then the approximation of the MMS produced by RECE1 approaches $\alpha$. 
{We explain here how we handle all values of $n$ in the special case that $n_i = n$ for every $i$.} Select $\delta$ and $n_0$ that maximize $\min[(1 - \delta)\alpha, \frac{1}{n_0}]$ (this last value will serve as our $\rho$) subject to the constraint that $n_0$ is the largest integer satisfying {$e^{-\Omega(\delta^2 {n_0})} \ge \frac{1}{n_0}$}. For all $n > n_0$ the bounds of Lemma~\ref{lem:deviation} allow us to use a union bound over all agents, and infer that with positive probability every agent gets at least $(1 - \delta)\alpha$-MMS. For all $n \le n_0$, we use Proposition~\ref{pro:EF1subadditive} to deduce $\frac{1}{n}$-MMS, which is at least $\frac{1}{n_0}$-MMS. 
Handling all values of $n$ when the $n_i$ might be smaller than $n$ uses principles similar to the above and gives the same value of $\rho$. 
See details in {Section~\ref{subsec:constant-all-nprime}}.

Before explaining key ideas in the proofs of our main lemmas, we observe that Lemma~\ref{lem:deviation} seems to fail the following elementary sanity check. On average, each agent receives $\frac{m}{n}$ items. If we consider a sequence of instances with growing $n$ in which the ratio $\frac{m}{n}$ remains the same, say $d$, then a typical agent gets roughly $d$ items, independently of the value of $n$. As these $d$ items have random values (because we have a random order over the items), one might expect that the probability of large deviations from the expected value will be exponentially small in $d$. So, how could it possibly be exponentially small in {$n_i$ (which can be as large as $n$), when the number of received items does not depend on $n_i$?} 

The answer to the above sanity check is the following. Though it is true that on average an agent receives $d$ items, the ECE algorithm may spread the items unevenly among the agents. In particular, the number of items that an agent receives may deviate drastically from $d$ (by a factor that depends on $n$) if there is a need for it, breaking the argument that with only $d$ items concentration results can depend only on $d$. As a concrete example, $d < 2$ in Example~\ref{ex:1}, but ECE1 gives one of the agents $n$ items.

\subsection{Key ideas in our proofs}

Here we explain the key ideas that we use in order to prove our two main lemmas. We also mention technical refinements that are used in order to implement these ideas. For simplicity and without loss of generality, we assume that $MMS(\items, v_i, n) = 1$.

{\bf Key idea 1.} It is very difficult to analyse $v_i(A_i)$ directly. So instead, we analyse a proxy for $v_i(A_i)$. This proxy is $W_i = \sum_{j \in [n]} v_i(A_j)$. It can serve as a proxy to $v_i(A_i)$, because due to the EFL property, we can lower bound $v_i(A_i)$ by $\frac{1}{2n-1} W_i$. (This statement is not accurate.  For a more accurate statement, see the technical refinements below.) The use of $W_i$ helps in both lemmas. For Lemma~\ref{lem:expectation}, we are now in the situation in which every item contributes to $W_i$ (though it might not be easy to figure out what this contribution is), avoiding the need to try to analyse the composition of $A_i$. For Lemma~\ref{lem:deviation}, as $W_i$ depends on all items (and we may assume that $m \ge n$, as otherwise the MMS is~0), Lemma~\ref{lem:deviation} now passes the sanity check of having the probability depend on $n$.

{\bf Technical refinement 1a.} Recall that $W_i = \sum_{j \in [n]} v_i(A_j)$. However, we remove from this sum those bundles $A_j$ for which their first item $e$ satisfies $v_i(e) \ge \alpha$. In ECE, such bundles cannot receive any additional item, unless agent $i$ already holds a bundle of value at least  $\alpha$. By removing these bundles, we will be able to assume that every item has value at most $\alpha$ (in fact, strictly smaller).

{{\bf Technical refinement 1b.} For our analysis of $W_i$, we may assume that for every item $e$, the bundle $A_j$ receiving it had {value at most $\alpha$ under $v_i$} at the time. This is because otherwise, agent $i$ already holds a bundle of value at least $\alpha$. To incorporate this assumption into our analysis, we replace $W_i$ by $W'_i$, defined via per-item contributions. {(Important remark. In our actual formal definitions, as appearing in Section~\ref{subsec:proxies}, $W'_i$ is defined with respect to a modified valuation $v'_i$ rather than the original $v_i$. The reason for this will be explained in key {idea~4} 
below. However, for this part of the overview, we pretend that {$W'_i$} is defined with respect to $v_i$.)} Each item $e$ contributes its marginal value {$v_i(e \mid A_e^{<e})$} if the receiving bundle had {$v_i$ value} less than $\alpha$ at the time, and contributes its full value {$v_i(e)$} otherwise.  If any bundle crosses the threshold $\alpha$, the EFL property already ensures {$v_i(A_i) \ge \alpha$}, so the ``otherwise'' case is benign.} 
As a result of key idea 1 and of its technical refinements, we have the following proposition.

\begin{proposition}
    \label{pro:proxyforv}
    Let {$n_i$} denote the number of bundles that remain after technical refinement 1a. Then $v_i(A_i) \ge \min[\alpha, \; \frac{W'_i}{2n_i - 1}]$.
\end{proposition} 

The proof of Proposition~\ref{pro:proxyforv} appears in {Section~\ref{subsec:proxies}}. 
Thus, we shall want to show that $W'_i \ge (2n_i - 1) \alpha$. This involves computing the expectation of $W'_i$ (implying Lemma~\ref{lem:expectation}), and proving concentration results for $W'_i$ (implying Lemma~\ref{lem:deviation}). {(Important technical note. The proof of Lemma~\ref{lem:expectation} actually uses $W_i$, not $W'_i$. Modifying $W_i$ to $W'_i$ is done for the purpose of proving Lemma~\ref{lem:deviation}. )}

{\bf Key idea 2.} Consider an MMS partition $B_1, \ldots, B_n$ for agent $i$. For each $B_j$, let $C_j$ denote the total contribution of its items to $W'_i$. A key observation that we make is that we can lower bound the expectation of each $C_j$ only based on the fact that the order of items within the respective $B_j$ is random, and independently of anything else (for lower bounding the expectation, there is no need to try to infer anything about individual values of items, the rounds in which they were chosen or the bundles which they entered). This is done using the following principles.  Suppose without loss of generality that $MMS(\items, v_i, n) = 1$ and let $k$ denote the initial number of items in $B_j$. Then when only $r$ items remain in $B_j$, the expectation for the total value that remains in $B_j$ is at least $\frac{r}{k}$. This follows easily from the random order of items in $B_j$, and from submodularity of {$v_i$}. Now, one of these items $e$ is selected at random, and {is allocated to some bundle $A_{\ell}$. If {$v_i(A_\ell)$} $< \alpha$ at that time, the contribution of $e$ to $W'_i$ is the marginal value. The expected contribution of this random $e$ to $W'_i$ is at least} $\frac{1}{r}(\frac{r}{k} - \alpha) = \frac{1}{k} - \frac{\alpha}{r}$. Also here, this follows easily from the random order of items in $B_j$, and from submodularity of {$v_i$}. 
Summing over all $k$ items of $B_j$, we get a lower bound on the expected value of $C_j$. This lower bound depends on $k$, but only to a small extent, and we can easily derive from it a lower bound that applies to all possible values of $k$. Summing over all $B_j$ (which simply means multiplying this lower bound by $n$), we get a lower bound on the expected value of {$W'_i$}. 

{\bf Technical refinement 2a.} As we consider RECE1 rather than RECE, the order of items within $B_j$ is only partly random. The first item to enter each $A_{\ell}$ may be arbitrary (in RECE1 it is chosen greedily according to $v_\ell$, and $v_{\ell}$ may have an arbitrary item as its most valuable item), and if one or more of these first items come from $B_j$, then the order over these items is not random. {However, each of these first items contributes its full $v_i$ value to $W'_i$ (as the receiving bundle was previously empty, so the contribution equals the full value).} Consequently, it is not hard to show that our lower bound on the expectation of $C_j$ still applies.

As a result of key idea 2 and of its technical refinement, we have the following lemma.

\begin{lemma}
\label{lem:bundleContribution}
Suppose that bundle $B_j$ in the MMS partition for $v_i$ does not contain any item $e$ satisfying $v_i(e) \ge \alpha$.  Then, its expected contribution to $W'_i$ satisfies:
\[
\kaliEXPECTED[C_{j}]\ \geq\ {1-\alpha-\alpha\ln\left(\frac{1}{\alpha}\right).}
\]
\end{lemma}

{The proof of Lemma~\ref{lem:bundleContribution} appears in {Section~\ref{subsec:avg-complement}}.} Lemma~\ref{lem:bundleContribution} suffices in order to prove Lemma~\ref{lem:expectation} {(see Section~\ref{lemm:expected_proof})}. Thus, it remains to prove Lemma~\ref{lem:deviation}.

{{\bf Key idea~3.} We would like to show that the value of the random variable $W'_i$ is very unlikely to be much smaller than its its expectation.
This seems plausible, as $W'_i$ is a sum of $m$ per-item contributions, and the contribution of each item is bounded (by $\alpha$). Thus, it seems natural to define a Doob martingale whose initial value is $E[W'_i]$ and whose final value is $W'_i$, and to apply standard bounded-difference martingale concentration results. However, more careful inspection of this martingale shows that the bounded difference property does not hold. }
Changing the item allocated in round~$r$ can dramatically change later rounds, so the effect of a single round on the final value of $W'_i$ is not bounded. Our next key idea 
is to lower bound the above Doob martingale by a submartingale, that we shall refer to as $\Lambda_0, \ldots, \Lambda_m$. Unlike the standard Doob martinagle,  $\Lambda_0$ will not be equal to $E[W'_i]$, but rather to our lower bound on $E[W'_i]$, the one implied by Lemma~\ref{lem:bundleContribution} and used in our proofs. In later steps $r$, the value $\Lambda_r$ will be the sum of two terms: the value accumulated into $W'_i$ so far, and a lower bound on the expectation of the additional value that $W'_i$ will accumulate in future rounds. Thus, $\Lambda_m = W'_i$, as desired. We will show that the steps of the submartingale do have the bounded difference property: the value  cannot change in a single step by more than $\alpha$ (see part (a) of Lemma~\ref{lem:Mk-bounded-var}.)

{{\bf Key idea 4.} For the submartingale $\Lambda_0, \ldots, \Lambda_m$, we wish to show that $\Lambda_m \ge \Lambda_0 - o(n)$, with probability higher than $1 - \frac{1}{n}$. The standard Azuma inequality for large deviations of submartingales does not apply here, as it would bound the deviation as a function of $m$, not $n$. For this reason, we use a different concentration result for martingales, Freedman's inequality (Theorem~\ref{thm:freedman}). In our use of Freedman's inequality, we need a deterministic $O(n)$ upper bound on $W'_i$. We achieve this in two steps. Recall that we assume that the value of the MMS is~1. The first step is standard and easy: replace $v_i$ by a ``capped" submodular valuation $\hat{v_i} = \min[v_i, 1]$. Consequently, each bundle in the MMS partition has value~1, and the sum of their values is $n$. However, this by itself does not upper bound $W'_i$ by $n$, because for submodular valuations, a single MMS bundle $B_j$ may contribute to $W_i$ (and thus to $W'_i$) much more than its value. This may happen if items in $B_j$ are substitutes of each other, thus contributing to $B_j$ marginal values that are much less than their stand-alone values, whereas in the allocation with respect to which $W'_i$ is computed, they are in different bundles and contribute their full stand-alone values. Our next key idea is to lower bound $\hat{v}_i$ by a submodular function $v'_i$ that satisfies $v'_i(B_j)=\hat{v}_i(B_j) = 1$, and importantly, is additive on each~$B_j$. With respect to this $v'_i$, we indeed have a deterministic upper bound of $n$ on $W'_i$.
(This upper bound is used in item (b) of Lemma~\ref{lem:Mk-bounded-var}.)} 
Importantly, RECE1 is run with the true valuations~$v_i$, so the final allocation is EFL with respect to the true valuations; the replacement of~$v_i$ by~$v'_i$ is done only in the analysis of $W'_i$. Lemma~\ref{lem:MMSoverAdditive} states the properties of~$v'_i$ and shows that such a $v'_i$ is guaranteed to exist.

\begin{lemma}
\label{lem:MMSoverAdditive}
    Let $v$ be a submodular valuation function and let $B_1, \ldots B_n$ be disjoint subsets of $\items$. Then there is a submodular function $v'$ with the following properties:
    \begin{itemize}
        \item $v'(S) \le v(S)$ for every $S \subseteq \items$.
        \item $v'(B_j) = v(B_j)$ for each of the disjoint subsets $B_j$.
        \item $v'$ is additive on each $B_j$: for every $j\in[n]$ and every $S\subseteq B_j$, $v'(S)=\sum_{e\in S} v'(\{e\})$, and in particular $\sum_{e\in B_j} v'(\{e\})=v'(B_j)$.
    \end{itemize}
\end{lemma}

The proof of Lemma~\ref{lem:MMSoverAdditive} appears in Section~\ref{app:vprime}. {The existence of such a $v'$ that is XOS can be shown quite easily, but does not suffice for our intended use of $v'_i$ (we make use of Claim~\ref{clm:marginal-beta}, but this claim does not hold for XOS valuations). The difficulty {in the proof of Lemma~\ref{lem:MMSoverAdditive}} is to ensure that $v'$ is submodular, which our construction achieves.}

This completes our overview for the proof of Theorem~\ref{thm:EFLMMSsubmodular}. 
\subsection{Discussion}

{We introduced the randomized allocation algorithm RECE1. Its output is always an EFL allocation. Moreover, we showed that if agents have submodular valuations, its distribution over allocations is {$0.2219$-MMS} {\em ex-ante} (every agent gets in expectation a value of at least {$0.2219$-MMS}), and the support of this distribution includes allocations that are {$0.142$-MMS} {\em ex-post} (allocations in which every agent gets at least {$0.142$-MMS}). Our analysis of RECE1 is not tight, and it would be interesting to see if significantly better approximation ratios can be proved for it (both ex-ante and ex-post).}

Recall that our intended approach was to analyse BECE, and RECE1 serves as a lower bound for BECE. It would be desirable to find stronger lower bounds for BECE than those implied by RECE1, or alternatively, to be able to analyse those approximations to the MMS that RECE1 offers with only exponentially small probability. This seems necessary if one wishes to extend our results to XOS valuations (and later, also to {subadditive} valuations), because RECE1 has at best exponentially small probability of offering a constant factor approximation for the MMS in this case. 

Finally, note that there are EF1 allocations that cannot be obtained by any version of the ECE algorithm. For example, consider two agents and seven items. Agents have additive valuations, where $v_1(e_1) = v_1(e_2) = 8$ and $v_2(e_1) = v_2(e_2) = 13$, {and every} other item has value~3 for each of the agents. The allocation $A$ that gives {$A_1 = \{e_1, e_2\}$ to agent~1  and $A_2 = \{e_3, \dots, e_7\}$ to agent~2 is EF1 (and EFL). ECE cannot reach this allocation. Removing a single item from $A_1$, the allocation is EF1 only if agent~2 is the one who holds the bundle $\{e_2\}$, and after receiving also the removed item, agent~2 would not agree to swap bundles with agent~1. Likewise, removing a single item from $A_2$, the allocation is EF1 only if agent~2 is the one who holds the bundle $A_1$, and after inserting the removed item back to $A_2$, agent~2 would not agree to swap bundles with agent~1}. Hence, no order of the items in ECE can result in the allocation $A$. This shows that ECE cannot be used to generate all possible EF1 allocations, suggesting that the best approximation for MMS offered by EF1 allocations might not be achievable by BECE. Thus, it may be fruitful to develop other approaches for producing allocations that are simultaneously EF1 and approximate MMS.

\subsection{Additional related work} \label{sec:related}

Envy-freeness (EF) is a classical comparison-based fairness notion
\cite{Foley67, Varian74}. For indivisible goods, EF may fail to exist, motivating relaxations.
{\citet{conf/sigecom/LiptonMMS04} proved that an EF1 allocation exists for every profile of monotone valuations, and gave the envy-cycle elimination algorithm that computes an EF1 allocation.} {The stronger relaxation, envy-freeness up to any good (EFX), has been introduced in~\cite{journals/teco/CaragiannisKMPS19}. It is known to exist in special cases, such as if all agents have the same valuation function~\cite{journals/siamdm/PlautR20}, {or three agents with additive valuations~\cite{journals/jacm/ChaudhuryGM24}.}
{For submodular valuations, however, it was recently shown that EFX allocations need not exist~\cite{journals/corr/abs-2604-18216, mackenzie2026counterexamplesefxsubmodularsubadditive}.} Another strengthening of EF1 is envy-freeness up to one less-preferred good (EFL), introduced in {\citet{conf/aaai/BarmanBMN18}, who prove that EFL allocations always exist for additive valuations}, and as noted in~\cite{journals/corr/abs-2505-19961}, this extends to general monotone valuations.

The maximin share (MMS) was proposed by Budish~\cite{budish2011combinatorial}. Exact MMS allocations may not exist even for
additive valuations~\cite{journals/jacm/KurokawaPW18}, motivating approximation guarantees. For additive valuations, {Kurokawa et al.}
proved the existence of $2/3$-MMS allocations.
{A long list of subsequent work (such as~\cite{journals/talg/AmanatidisMNS17, journals/teco/BarmanK20, conf/ijcai/AkramiGST23, conf/soda/AkramiG24, conf/soda/HeidariKSS26, journals/corr/abs-2511-13056}) led to an improved
$\rho=7/9$ approximation.}
On the impossibility side, there are instances with additive
valuations in which no allocation achieves $\rho>39/40$~\cite{conf/wine/FeigeST21}.

For monotone submodular valuations, MMS approximation guarantees {that were obtained} include $0.21$-MMS in~\cite{journals/teco/BarmanK20}, $1/3$-MMS in~\cite{conf/sigecom/GhodsiHSSY18} and $10/27$-MMS in~\cite{journals/corr/abs-2303-12444}.

{Some} works consider {approximate MMS allocation when} $n$ is bounded.
{For $n=3$ agents with additive valuations, following a sequence of works~\cite{journals/talg/AmanatidisMNS17, journals/tcs/GourvesM19, journals/corr/abs-2205-05363}, the current best approximation ratio is $11/12$-MMS. For $n=4$ agents, $1/2$-MMS allocations exist for subadditive (and hence also submodular) valuations~\cite{conf/ijcai/0001CMS25}.}}

{A line of work studies combinations of EF-style relaxations (e.g., EF1/EFX/EFL) with share-based guarantees such as $\rho$-MMS.
For additive valuations, allocations that are simultaneously EF1 and $4/7$-MMS were obtained using a variant of GECE~\cite{journals/tcs/AmanatidisMN20}, and more recently, $2/3$-MMS with EF1 guarantees were achieved~\cite{conf/aaai/AkramiR25a}. {For the broader class of restricted MMS-feasible valuations, \cite{conf/sigecom/AshuriG25} show the existence of allocations that simultaneously satisfy EFL and MXS; the latter in particular implies $4/7$-MMS.} {A result that unifies earlier results and applies to all monotone valuations is allocations that are simultaneously EFL and RMMS (the {\em residual maximin share})~\cite{journals/corr/abs-2505-19961}, which for additive valuations imply $\frac{2}{3}$-MMS.}
Beyond additive valuations, fewer such combined results are known.
For subadditive valuations, EF1 allocations offering at least half of the maximum Nash social welfare are shown to exist~\cite{conf/soda/BarmanS26}.}

{\cite{conf/wine/BuLLST24} studied fair division of indivisible goods with comparison-based queries, and presented results for agents with additive valuations. Our RECE1 allocation algorithm gives constant approximate MMS guarantees within this query model, when valuations are submodular. Comparison queries between disjoint bundles suffice for RECE1. For submodular valuations, such queries are strictly less expressive than comparisons between possibly overlapping bundles. }

{\section{Preliminaries}} \label{sec:perliminaries}

\subsection{Fair Allocation Setting}

{Recall from the Introduction that $\agents=\{1,\dots,n\}$ is the set of agents, $\items$ is the finite set of indivisible items, and each $v_i:2^{\items}\to\mathbb{R}_{\ge 0}$ is normalized and monotone. An \emph{allocation} is an $n$-tuple $\mathbf{A}=(A_1,\dots,A_n)$ of pairwise-disjoint subsets of $\items$ with $\bigcup_{i=1}^n A_i = \items$.}

{We will use the following well-known monotonicity property of the MMS value, which holds for every monotone valuation.}

\begin{claim}\label{clm:not_too_small_mms}
For every monotone valuation function $v_i$, item set $\items$, and every subset $S\subseteq \items$ with $|S|=k<n$,
\[
MMS(\items\setminus S, v_i, n-k)\ \ge\ MMS(\items, v_i, n).
\]
\end{claim}

{The claim follows easily from the fact that at least $n-k$ bundles of the original MMS partition for $v_i$ do not contain any item from $S$.}

{
\subsection{Submodular valuations}

{Recall from Definition~\ref{def:valuation-classes} the definition of \emph{submodular} valuations (diminishing marginal values), and that every submodular valuation is subadditive.}

{We use two well-known properties of monotone submodular valuations.}

{\begin{claim}[Chain-marginal additive lower bound]\label{clm:chain-marginal-lower-bound}
Let $v$ be a monotone submodular valuation, let $B\subseteq\items$, and fix any ordering $e_1,e_2,\ldots,e_k$ of the items in $B$. Define the chain prefixes $P_r:=\{e_1,\ldots,e_r\}$ (with $P_0:=\emptyset$) and the chain-order marginals
\[
w(e_r)\ :=\ v(e_r\mid P_{r-1}),\qquad r=1,\ldots,k.
\]
Then the additive function $w(e)$ is a pointwise lower bound on $v$ over subsets of $B$:
\[
\sum_{e\in S} w(e)\ \le\ v(S)\qquad\text{for every }S\subseteq B,
\]
with equality on $S=B$.
\end{claim}
}

{The second is the standard averaging bound: a uniformly random subset of a bundle has expected value at least proportional to its size.}

{
\begin{claim}[Averaging bound for submodular valuations]\label{clm:avg-submodular}
Let $v$ be a monotone submodular valuation, let $B$ be a set of $k$ items, and let $S \subseteq B$ be a uniformly random subset of size $r$ (where $0 \le r \le k$). Then
\[
\mathbb{E}[v(S)] \ge \frac{r}{k}\,v(B).
\]
\end{claim}
The proof appears in Appendix~\ref{app:expected_lower_sub}.
}

}
{
\subsection{Martingales and Freedman's inequality}\label{subsec:prob-tools}
}

A \emph{filtration} $\{\mathcal{F}_k\}_{k=0}^m$ is a nested sequence of $\sigma$-algebras $\mathcal{F}_0 \subseteq \mathcal{F}_1 \subseteq \cdots \subseteq \mathcal{F}_m$, representing increasing information over time.
A sequence of random variables $\{X_k\}_{k=0}^m$ adapted to $\{\mathcal{F}_k\}$ is a \emph{martingale} if $\kaliEXPECTED[X_k \mid \mathcal{F}_{k-1}] = X_{k-1}$ for all $k$, and a \emph{submartingale} if $\kaliEXPECTED[X_k \mid \mathcal{F}_{k-1}] \ge X_{k-1}$.

We use a concentration inequality for submartingales known as Freedman's inequality, which exploits the predictable quadratic variation rather than worst-case bounded increments.

\begin{theorem}[Freedman's inequality (one-sided)~\cite{freedman1975tail, fan2015exponential}]\label{thm:freedman}
{Let $\{Z_k\}_{k=0}^m$ be a submartingale with respect to a filtration $\{\mathcal{F}_k\}$, with $Z_0 = 0$ and increments $D_k := Z_k - Z_{k-1}$ satisfying $D_k \ge -c$ a.s.\ for all $k$.
Define the quadratic variation $V_m := \sum_{k=1}^m D_k^2$.
Then for every $t,v^2 > 0$, 
\[
\Prb\left[Z_m \le -t \text{ and } V_m \le v^2\right]
\ \le\ \exp\left(-\frac{t^2}{2(v^2 + ct)}\right).
\]
}

\end{theorem}

\section{The RECE1 Algorithm} \label{sec:RECE1}

\subsection{The original {\em envy cycle elimination} (ECE) algorithm}
\label{sec:ECE}

We first recall the ECE algorithm of~\cite{conf/sigecom/LiptonMMS04}.
{We begin with the notion of the \emph{envy graph}.}
\begin{definition} \label{def:envy_graph}
    For a partial allocation $\hat{A}$, the envy graph \(G(\hat{A})\) {is the directed graph with} vertex set \(V=\agents\).
    For vertices $i,j\in V$, we include a directed edge \((i,j)\in E\) if
    \[
     v_i(\hat{A}_j) > v_i(\hat{A}_i).
    \]
    Equivalently, we say that \(i\) \emph{envies} \(j\) when $v_i(\hat{A}_j) > v_i(\hat{A}_i)$.
\end{definition}

{The following claim shows that any partial allocation can be modified so as to produce a vertex with no incoming edges in the envy graph.}

\begin{claim} \label{clm:envy_cycle_sink}
    For every (partial) allocation $\hat{A}$ {one can obtain a new (partial) allocation $\hat{A}'$, by permuting bundles among agents, with the following properties:}
    \begin{itemize}
        \item For every agent $i$, $v_i(\hat{A}'_i)\ge v_i(\hat{A}_i)$.
        \item {The envy graph $G(\hat{A}')$ is acyclic (and in particular has a vertex with no incoming edges).}
        {\item If $\hat{A}$ is EF1, then so is $\hat{A}'$.}
    \end{itemize}
\end{claim}

For completeness, we present the proof of this claim.

{We first define the cycle rotation process for a partial allocation $\hat{A}$.}
\begin{definition}[Cycle Rotation Process] \label{def:cycle_rot}

While $G(\hat{A})$ contains a directed cycle {$C=(i_0,i_1,\dots,i_{k-1},i_0)$}: 
\begin{itemize}
    \item For every agent $j\notin C$, set $\hat{A}'_j=\hat{A}_j$.
    \item For every {$\ell\in\{0,\dots,k-1\}$}, set $\hat{A}'_{i_\ell}=\hat{A}_{i_{\ell+1}}$ (where indices are modulo $k$).
    \item Update $\hat{A}\gets \hat{A}'$.
\end{itemize}
\end{definition}

\begin{proof} [Proof of Claim \ref{clm:envy_cycle_sink}]

By the definition of $G(\hat{A})$, along a directed edge $(i_\ell,i_{\ell+1})$ we have
$v_{i_\ell}(\hat{A}_{i_{\ell+1}})>v_{i_\ell}(\hat{A}_{i_\ell})$; hence after the reassignment on the cycle,
every agent on the cycle strictly increases her value for her own bundle, and agents outside the cycle are unchanged. 
Therefore, the process cannot continue indefinitely {(since bundles are only reassigned---not modified---each of the $n$ agents can increase her value at most $n-1$ times, so the process terminates in at most $n(n-1)$ rotations)}, and it terminates at an allocation $\hat{A}'$ whose envy graph is acyclic.
Every finite directed acyclic graph has a vertex with no incoming edges, {establishing the first two properties. For the third, if $\hat{A}$ is EF1, then for any pair $(i,j)$ the bundle $\hat{A}'_j$ equals $\hat{A}_k$ for some~$k$, and there exists $g \in \hat{A}_k$ with $v_i(\hat{A}_i) \ge v_i(\hat{A}_k \setminus \{g\})$. Since $v_i(\hat{A}'_i) \ge v_i(\hat{A}_i)$, we get $v_i(\hat{A}'_i) \ge v_i(\hat{A}'_j \setminus \{g\})$, so $\hat{A}'$ is also EF1.}
\end{proof}

Now we introduce the full algorithm.
\paragraph{ECE Algorithm} \label{alg:ECE}
\begin{itemize}
  \item Initialize $\hat{A}_i\gets \emptyset$ for all $i\in\agents$.
  \item While there exists an unallocated item $e\in \items\setminus \bigcup_{k=1}^n \hat{A}_k$:
  \begin{enumerate}
    \item Choose an arbitrary agent $i$ such that her vertex in $G(\hat{A})$ has no incoming edges.
    \item Choose an arbitrary item $e\in \items\setminus \bigcup_{k=1}^n \hat{A}_k$.
    \item Allocate $e$ to $i$'s bundle: $\hat{A}_i\gets \hat{A}_i\cup\{e\}$.
    \item {Update $\hat{A}$ using the Cycle Rotation Process (Claim~\ref{clm:envy_cycle_sink}).} 
  \end{enumerate}
  \item Output the full allocation $A=\hat{A}$.
\end{itemize}

{The following proposition is proved in~\cite{conf/sigecom/LiptonMMS04}.}
\begin{proposition} \label{prp:ECE_EF1}
    The full allocation $A$ produced by ECE is EF1.
\end{proposition}

{
For completeness, the proof of Proposition~\ref{prp:ECE_EF1} appears in Section~\ref{app:ECE_EF1} in the appendix.
}

\subsection{The RECE1 Algorithm}

We now define RECE1, which is a version of ECE in which {each agent's first received item is chosen greedily}, and the remaining items are allocated in a uniformly random order.

\paragraph{RECE1 Algorithm} \label{alg:RECE1}
\begin{enumerate}
    {\item \textbf{Greedy choice phase ({rounds} $t=1,\dots,n$).}
    Initialize $\hat{A}_i\gets\emptyset$ for all $i\in\agents$, and let $U\gets \agents$ be the set of agents that have not yet received an item.
    For each round $t=1,2,\dots,n$:
    \begin{enumerate}
        \item Let $R \gets \items\setminus \bigcup_{k=1}^n \hat{A}_k$ be the set of unallocated items.
        \item {Choose an arbitrary agent $i_t \in U$.}
        {\item Let $e_t \in \arg\max_{e\in R}\, v_{i_t}(e)$.}
        \item Allocate $e_t$ to agent $i_t$: $\hat{A}_{i_t}\gets \{e_t\}$, and update $U\gets U\setminus\{i_t\}$.
    \end{enumerate}}

    {Note that $G(\hat{A})$ is acyclic after the greedy phase (each agent's greedy choice ensures she does not envy any later-choosing agent), so step {(2)}(a) below is well-defined {for $t = n+1$.} }

    \item \textbf{Random-item phase (rounds $t=n+1,\dots,m$).}
     While {$\items\setminus \bigcup_{k=1}^n \hat{A}_k \neq \emptyset$}:
    \begin{enumerate}
        \item Choose an arbitrary agent $i$ such that her vertex in $G(\hat{A})$ has no incoming edges.
        \item Choose an item $e_t \in \items\setminus \bigcup_{k=1}^n \hat{A}_k$ uniformly at random.
        \item Allocate $e_t$ to $i$'s bundle: $\hat{A}_i\gets \hat{A}_i\cup\{e_t\}$.
        \item  {Update $\hat{A}$ using the Cycle Rotation Process (Claim~\ref{clm:envy_cycle_sink}).} 
        \item Update $t\gets t+1$.
    \end{enumerate}

    \item Output the full allocation $A=\hat{A}$.
\end{enumerate}

\begin{proposition} \label{prop:RECE1_EFL_as}
The allocation produced by RECE1 is EFL.
\end{proposition}

{Proposition~\ref{prop:RECE1_EFL_as} holds for every ECE variant in which each agent's first item is chosen greedily.} {See for example~\cite{conf/aaai/BarmanBMN18}, {which proves the claim in the additive case}. For completeness, we present a proof of Proposition~\ref{prop:RECE1_EFL_as} in Section~\ref{app:RECE1_EFL_as} in the appendix.}

\section{A Lower Bound on the Expected Welfare} \label{sec:W_i_expected}

{In this section we prove a lower bound on the expected fraction of the MMS received by each agent.} First,
as is standard in similar settings, scaling of the valuation function allows us to assume that $MMS_i = 1$ for every agent $i$. \begin{definition} [Scaling]
    For submodular valuation function $v$ with MMS value $\gamma > 0$,  define the scaled function $v'$ as the function satisfying $v'(S)= \frac{v(S)}{\gamma}$ for every set $S \subseteq \items$.
\end{definition}

We may use the scaled function instead of $v_i$, because the property of an allocation being $\alpha$-MMS and/or EFL is not affected by multiplicative scaling of valuations.

As explained in Section~\ref{sec:introduction}, we introduce functions $W_i$ and $W'_i$ that assist in the analysis of $v_i(A_i)$.

\subsection{The functions $W_i$ and $W_i'$}\label{subsec:proxies}

Let $A$ be the \emph{final} allocation output by RECE1 (so $A=\hat{A}$ at termination).
Define
\[
W_i\ :=\ \sum_{j=1}^n v_i(A_j).
\]
We refer to $W_i$ as the {\em welfare} {with respect to $v_i$.} Next, we define an auxiliary function $W'_i$ as follows.

{After scaling so that $MMS_i = 1$, we {apply} two transformations. First, define the capped valuation $\hat v_i(S):=\min\{v_i(S),1\}$, so that each MMS bundle has $\hat v_i$-value exactly~$1$; truncating a nonnegative monotone submodular valuation at a constant preserves monotonicity and submodularity, so $\hat v_i$ is monotone submodular. Second, apply Lemma~\ref{lem:MMSoverAdditive} to $\hat v_i$ and the MMS partition $B_1,\ldots,B_n$ to obtain a monotone submodular function $v'_i$ satisfying: (i)~$v'_i(S) \le \hat v_i(S) \le v_i(S)$ for every $S \subseteq \items$; (ii)~$v'_i(B_j)=\hat v_i(B_j)= 1$ for each MMS bundle $B_j$; and (iii)~$v'_i$ is additive within each $B_j$. We now define a {proxy welfare}.}

{\begin{definition}[Per-item contribution and $W'_i$]\label{def:tilde-v}
Fix agent $i$ and the allocation produced by RECE1.
For each item $e\in \items$, let $A_e$ denote the bundle receiving $e$,
and let $A_e^{<e}$ denote the contents of $A_e$ just before $e$ is allocated.
The \emph{contribution} of $e$ to $W'_i$ is
\[
c_i(e)\ :=\
\begin{cases}
{v'_i(e \mid A_e^{<e})} & \text{if }{v'_i(A_e^{<e})} < \alpha, \\[4pt]
{v'_i(e)} & \text{if } {v'_i(A_e^{<e})} \ge \alpha.
\end{cases}
\]
Define
\[
W_i'\ :=\ \sum_{e\in \items} c_i(e).
\]
\end{definition}}

\begin{observation}\label{obs:block-not-opening}
{By submodularity, {$c_i(e) \ge v'_i(e \mid A_e^{<e})$} always.
Moreover, {$c_i(e) \le v'_i(e)$} always.}
\end{observation}

\begin{claim}\label{clm:Wi-ge-Wiprime-or-alpha}
For every run of RECE1, either $v_i(A_i) \ge \alpha$ or $W_i \ge W_i'$.
\end{claim}

\begin{proof}
{If $W_i' > W_i$, then since $W_i = \sum_j v_i(A_j) \ge \sum_j v'_i(A_j)$ (as $v'_i \le v_i$), there must exist some item~$e$ with $c_i(e) > v'_i(e \mid A_e^{<e})$. By Definition~\ref{def:tilde-v}, this happens only when $v'_i(A_e^{<e}) \ge \alpha$, hence $v_i(A_e^{<e}) \ge v'_i(A_e^{<e}) \ge \alpha$. That is, some bundle had $v_i$-value $\ge \alpha$ and received another item. At that round, agent~$i$ did not envy the recipient (who was chosen as a sink in the envy graph), so $v_i(A_i) \ge \alpha$.}
\end{proof}

\subsection{Lower bounding the expected welfare proxy $W'_i$}\label{subsec:avg-complement}

The following fact is used in our analysis:

\begin{claim}\label{clm:marginal-beta}

Consider submodular valuation function $v$, and bundles $A,B \subseteq \items$, s.t.  {$v(A)\le \alpha$} and $v\left(B\right)=\beta$.
{Suppose $e$ is drawn uniformly at random from $B$ and added to $A$.
Then the expected marginal value of $e$ added to $A$ satisfies}

\[
{\kaliEXPECTED_{e\sim\mathbb{U}(B)}}\left[v\left(e\mid A\right)\right]\geq \max\Big\{\left(\beta-\alpha\right)\cdot\frac{1}{\left|B\right|},0\Big\}.
\]
\end{claim}
The proof is provided in {Appendix \ref{app:marginal-beta}}.

Let $e$ be an item allocated by RECE1 in the random-item phase, and let $A_e$ denote the bundle in which RECE1 places $e$. 

{The contribution $c_i(e)$ of item $e$ to $W'_i$ is given by Definition~\ref{def:tilde-v}. In particular, if {$v'_i(A_e^{<e}) < \alpha$}, then {$c_i(e) = v'_i(e \mid A_e^{<e})$}.}

Let $B_1, \ldots, B_n$ denote the bundles of the MMS partition with respect to {$v_i$ (and hence also {$v'_i$}).}
Next, we focus on the total contribution of a single MMS bundle $B_j$ to $W'_i$.

\begin{definition}\label{def:Cj}
Fix an MMS bundle $B_j$. Let $(e_1,\dots,e_{|B_j|})$ denote the items of $B_j$ {in the order in which} they are selected during RECE1. Define
\[
C_j\ :=\ {\sum_{r=1}^{|B_j|} c_i(e_r).}
\]
\end{definition}

Now we can see:

\[
W'_i=\sum_{j=1}^{n}C_{j}.
\]

\subsubsection{Lower bound on $\mathbb E[C_j]$ }\label{subsubsec:ECj}

{
\begin{definition}\label{def:g-alpha}
For $\alpha\in(0,1)$, define
\[
g(\alpha)\ :=\ {1-\alpha-\alpha\ln\Bigl(\frac{1}{\alpha}\Bigr).}
\]
\end{definition}
\begin{observation}\label{obs:g-ge-2alpha}
For every $\alpha \le 0.2219$, $g(\alpha) \ge 2\alpha$. 
\end{observation}
}

{Denote $x_+ := \max\{x,\,0\}$.} The following definition will be useful for our analysis:
{
\begin{definition}[Bundle potential]\label{def:P-pot}
For every $j\in[n]$ and every $R\subseteq B_j$, {the {\em potential} of $R$, denoted $P(R)$, is defined as:}

\[
P(R)\ :=\ \sum_{r=1}^{|R|}\frac{1}{r}\,\kaliEXPECTED_{S\sim\mathbb{U}\binom{R}{r}}\bigl[{(v'_i(S)-\alpha)_+}\bigr],\qquad P(\emptyset):=0.
\]
{Here $\binom{R}{r}$ denotes the family of $r$-element subsets of $R$, and $\mathbb{U}(\cdot)$ denotes the uniform distribution on the indicated set.}
\end{definition}

\begin{claim}[Per-round contribution]\label{clm:per-round-contrib}
Fix an MMS bundle $B_j$ and a non-empty subset $S\subseteq B_j$. Suppose during RECE1's random phase an item $e$ is drawn uniformly at random from $S$ and placed by the algorithm into some recipient bundle $A_e^{<e}$. Then

\[
\kaliEXPECTED_{e\sim\mathbb{U}(S)}\bigl[c_i(e)\bigr]\ \ge\ {\frac{(v'_i(S)-\alpha)_+}{|S|}},
\]
where $c_i$ is the contribution of Definition~\ref{def:tilde-v}.
\end{claim}

\begin{proof}
The case-split in Definition~\ref{def:tilde-v} is determined by $v'_i(A_e^{<e})$. If $v'_i(A_e^{<e})<\alpha$, then $c_i(e)=v'_i(e\mid A_e^{<e})$, and Claim~\ref{clm:marginal-beta} applied with $A:=A_e^{<e}$ and $B:=S$ (using subset-additivity of $v'_i$ on $B_j$, so $v'_i(S)=\sum_{e\in S}v'_i(\{e\})$) gives {$\kaliEXPECTED_{e\sim\mathbb{U}(S)}[v'_i(e\mid A_e^{<e})]\ge (v'_i(S)-\alpha)_+/|S|$.} If instead $v'_i(A_e^{<e})\ge\alpha$, then $c_i(e)=v'_i(\{e\})$, and applying Claim~\ref{clm:marginal-beta} with $A:=\emptyset$ and $B:=S$ yields {a bound even larger than $\phi(v'_i(S))/|S|$.}
\end{proof}

Intuitively, $P(R)$ aggregates the per-round bound of Claim~\ref{clm:per-round-contrib} over the $|R|$ rounds at which items from $R$ are revealed in a uniformly random order: when the size-$r$ revealed subset is $S$, the next item contributes at least {$(v'_i(S)-\alpha)_+/|S|$} in expectation, and summing these contributions for $r=1,\dots,|R|$ yields $P(R)$.

\begin{claim}[Recursion for $P$]\label{obs:P-recursion}
For every $R\subseteq B_j$ with $R\ne\emptyset$,
\[
P(R)\ =\ {\frac{(v'_i(R)-\alpha)_+}{|R|}}\ +\ \kaliEXPECTED_{e'\sim\mathbb{U}(R)}\bigl[P(R\setminus\{e'\})\bigr].
\]
\end{claim}

\begin{proof}
Set $q:=|R|$. The $r=q$ term in Definition~\ref{def:P-pot} is {$(v'_i(R)-\alpha)_+/q$} (since $\binom{R}{q}=\{R\}$). For $1\le r\le q-1$, drawing $S$ uniformly from $\binom{R}{r}$ has the same distribution as first drawing $e'$ uniformly from $R$ and then drawing $S$ uniformly from $\binom{R\setminus\{e'\}}{r}$. Substituting and recognizing the resulting double sum as $\kaliEXPECTED_{e'\sim\mathbb{U}(R)}[P(R\setminus\{e'\})]$ proves the recursion.
\end{proof}

\begin{lemma}[Bundle bound]\label{lem:P-bundle-bound}
For every $j\in[n]$, $P(B_j)\ \ge\ g(\alpha)$.
\end{lemma}

\begin{proof}
Let $k:=|B_j|$. By Claim~\ref{clm:avg-submodular} applied to $v'_i$ on $B_j$ (which is additive on $B_j$ by Lemma~\ref{lem:MMSoverAdditive}), $\kaliEXPECTED_{S\sim\mathbb{U}\binom{B_j}{r}}[v'_i(S)]\ge r/k$. {Since $\kaliEXPECTED[\max\{X,Y\}]\ge \max\{\kaliEXPECTED[X],\kaliEXPECTED[Y]\}$ for any random variables $X,Y$, applying this with $X:=v'_i(S)-\alpha$ and $Y:=0$ gives $\kaliEXPECTED[(v'_i(S)-\alpha)_+]\ge (\kaliEXPECTED[v'_i(S)]-\alpha)_+\ge (r/k-\alpha)_+$.} Hence
\[
P(B_j)\ \ge\ \sum_{r=1}^{k}\frac{(r/k-\alpha)_+}{r}\ =\ \sum_{r=1}^{k}\Bigl(\frac{1}{k}-\frac{\alpha}{r}\Bigr)_+\ \ge\ g(\alpha),
\]
where the last inequality {will be proved in} Lemma~\ref{lem:g-sum-bound}. \end{proof}
}

{The next deterministic property of $P$ controls how it changes when an item is dropped from a subset.}

\begin{lemma}[Drop lemma]\label{lem:drop}
{For every $R\subseteq B_j$ and every $e\in R$,}
\[
{P(R)-v'_i(\{e\})\ \le\ P(R\setminus\{e\})\ \le\ P(R).}
\]
\end{lemma}

{The proof appears in Appendix~\ref{app:drop-lemma}.}

\begin{proposition} \label{prop:expected_c_j}
Fix an agent $i$ and let $B_1,\ldots,B_n$ be the bundles of an MMS partition for {$v'_i$}. Then, for every $j\in[n]$,
\[
{\kaliEXPECTED[C_{j}]\ \geq\ P(B_j)\ \geq\ g(\alpha).}
\]

\end{proposition}

{\begin{proof}
Recall that RECE1 first allocates to each agent the highest-value single item remaining (the greedy phase), and then randomly allocates an item in every round (the random phase). Denote by $G_j\subseteq B_j$ the subset of items of $B_j$ chosen in the greedy phase, and let $R_j^0:=B_j\setminus G_j$ be the items of $B_j$ that enter the random phase.

Let $C_j^{\mathrm{rnd}}:=\sum_{e\in R_j^0} c_i(e)$ be the contribution of the random-phase items of $B_j$ (so $C_j=\sum_{g\in G_j} c_i(g)+C_j^{\mathrm{rnd}}$). The greedy phase is deterministic, so $G_j$ and $R_j^0$ are deterministic sets. We claim
\begin{equation}\label{eq:Crnd-vs-P}
\kaliEXPECTED\bigl[C_j^{\mathrm{rnd}}\bigr]\ \ge\ P(R_j^0).
\end{equation}

To see this, let $r$ run over the rounds of the random phase at which an item from $B_j$ is drawn, and let $e_r$ be that item; conditioned on the size-$r$ subset $S\subseteq R_j^0$ that has been drawn from $B_j$ so far, $e_r$ is uniformly distributed on $S$. By Claim~\ref{clm:per-round-contrib}, for every realization of $A_{e_r}^{<e_r}$,
\[
\kaliEXPECTED_{e_r\sim\mathbb{U}(S)}\bigl[c_i(e_r)\bigr]\ \ge\ {\frac{(v'_i(S)-\alpha)_+}{|S|}}.
\]
Taking outer expectation over $S$ (uniform on $\binom{R_j^0}{r}$) and summing over $r=1,\dots,|R_j^0|$ yields~\eqref{eq:Crnd-vs-P}.

Next, every greedy item $g\in G_j$ enters an initially empty bundle, so $v'_i(A_g^{<g})=0<\alpha$ and hence $c_i(g)=v'_i(g\mid\emptyset)=v'_i(\{g\})$ deterministically. Therefore
\[
\kaliEXPECTED[C_j]\ =\ \sum_{g\in G_j} v'_i(\{g\})\ +\ \kaliEXPECTED[C_j^{\mathrm{rnd}}]\ \ge\ \sum_{g\in G_j} v'_i(\{g\})\ +\ P(R_j^0).
\]
Iterating the drop lemma (Lemma~\ref{lem:drop}) along the items of $G_j$ gives $P(B_j)\le P(R_j^0)+\sum_{g\in G_j} v'_i(\{g\})$. Combining with the bundle bound (Lemma~\ref{lem:P-bundle-bound}),
\[
\kaliEXPECTED[C_j]\ \ge\ P(R_j^0)+\sum_{g\in G_j} v'_i(\{g\})\ \ge\ P(B_j)\ \ge\ g(\alpha). \qedhere
\]
\end{proof}}

\begin{proposition}\label{prop:EWprime}
{For} every agent $i$,
\[
\mathbb E[W'_{i}]
=\sum_{j=1}^n \mathbb E[C_j]
\ \ge\
n\cdot {g(\alpha)}.
\]
\end{proposition}

\begin{proof}
{Immediate from $W'_i=\sum_j C_j$, linearity of expectation, and Proposition~\ref{prop:expected_c_j}.}
\end{proof}

\begin{proposition}\label{prop:EW-true}
{Let}
\[
\mu_i := \kaliEXPECTED[v_i(A_i)].
\]
If $\mu_i \le \alpha$, then
\[
\kaliEXPECTED[W_i] \ge n\cdot g(\alpha).
\]
\end{proposition}

{\begin{proof}[Proof sketch] The argument follows the proof of Proposition~\ref{prop:expected_c_j} with $W'_i$ and the case-split contribution $c_i(e)$ replaced by $W_i$ and the true marginal $d_i(e):=v_i(e\mid A_e^{<e})$. The key observation is that at every random-phase round, the recipient bundle $A_{e_r}^{<e_r}$ is a sink in the envy graph, so $X_r:=v_i(A_{e_r}^{<e_r})\le v_i(A_i)$ deterministically, and hence $\kaliEXPECTED[X_r]\le\mu_i\le\alpha$. Applying Claim~\ref{clm:marginal-beta} conditionally and taking outer expectation reproduces the per-round bound $\kaliEXPECTED[d_i(e_r)]\ge \max\{1/k-\alpha/r,0\}$ from the proof of Proposition~\ref{prop:expected_c_j}; the same calculation then yields $\kaliEXPECTED[\sum_{e\in B_j}d_i(e)]\ge g(\alpha)$ per MMS bundle, and summing over the $n$ bundles gives $\kaliEXPECTED[W_i]\ge n\cdot g(\alpha)$. Greedy items are handled exactly as in Proposition~\ref{prop:expected_c_j}. The full proof appears in Appendix~\ref{app:EW-true}.
\end{proof}}

\subsection{Proof of Lemma~\ref{lem:expectation}}\label{lemm:expected_proof}

{In the proof of Lemma~\ref{lem:expectation}, and only in this proof, we assume that the greedy phase of RECE1 selects the agent in each round uniformly at random from the unselected agents. This does not affect any other result in the paper: Proposition~\ref{prop:RECE1_EFL_as} and the chain culminating in Proposition~\ref{prop:EW-true} hold for any ordering of agents in the greedy phase.}

\begin{proof}
Let $\mu_i:=\kaliEXPECTED[v_i(A_i)]$. Sort items in non-increasing $v_i$-value: $v_i(\{e_1\})\ge v_i(\{e_2\})\ge\cdots\ge v_i(\{e_m\})$ (where $m=|\items|$). Define
\[
\beta\ :=\ \tfrac{1}{n-1}\sum_{r=1}^{n-1} v_i(\{e_r\}),
\]
the average $v_i$-value of the $n-1$ highest-valued items.

\textit{Claim 1: $\mu_i\ge\tfrac{n-1}{n}\beta$.}
Item additions to agent $i$'s bundle and cycle-rotation steps each only increase $v_i(A_i)$ (Claim~\ref{clm:envy_cycle_sink}), so $v_i(A_i)\ge v_i(\{e_{i,\mathrm{gr}}\})$ pointwise, where $e_{i,\mathrm{gr}}$ is the item that agent~$i$ receives in the greedy phase. Let $r_i\in\{1,\ldots,n\}$ denote agent~$i$'s position in the random greedy ordering; $\Pr(r_i=r)=1/n$ for each~$r$. At position $r$, exactly $r-1$ items have been allocated; in particular, at most $r-1$ of the items $\{e_1,\ldots,e_r\}$ have been removed, so at least one of them remains, and the greedy rule selects an item of $v_i$-value $\ge v_i(\{e_r\})$. Therefore
\[
\kaliEXPECTED[v_i(\{e_{i,\mathrm{gr}}\})]\ \ge\ \tfrac{1}{n}\sum_{r=1}^{n}v_i(\{e_r\})\ \ge\ \tfrac{1}{n}\sum_{r=1}^{n-1}v_i(\{e_r\})\ =\ \tfrac{n-1}{n}\beta,
\]
and $\mu_i\ge\kaliEXPECTED[v_i(\{e_{i,\mathrm{gr}}\})]\ge\tfrac{n-1}{n}\beta$.

\textit{Claim 2: $\mu_i\ge\bigl(\kaliEXPECTED[W_i]-(n-1)\beta\bigr)/n$.}
Since the allocation $A$ is EFL (Proposition~\ref{prop:RECE1_EFL_as}), for every $j\ne i$ there exists $e_j^\star\in A_j$ with $v_i(A_j\setminus\{e_j^\star\})\le v_i(A_i)$, hence by subadditivity $v_i(A_j)\le v_i(A_i)+v_i(\{e_j^\star\})$. The $n-1$ items $\{e_j^\star:j\ne i\}$ are pairwise distinct (the bundles $A_j$ are disjoint) and lie outside $A_i$, so
\[
\sum_{j\ne i}v_i(\{e_j^\star\})\ \le\ \sum_{r=1}^{n-1}v_i(\{e_r\})\ =\ (n-1)\beta.
\]
Therefore $W_i=v_i(A_i)+\sum_{j\ne i}v_i(A_j)\le n\,v_i(A_i)+(n-1)\beta$. Taking expectations and rearranging gives the claim.

\textit{Combining.} Adding Claims~1 and~2,
\[
2\mu_i\ \ge\ \tfrac{n-1}{n}\beta\ +\ \tfrac{\kaliEXPECTED[W_i]-(n-1)\beta}{n}\ =\ \tfrac{\kaliEXPECTED[W_i]}{n}.
\]
Now fix any $\alpha\in(0,1)$. Either $\mu_i>\alpha$, or $\mu_i\le\alpha$; in the latter case Proposition~\ref{prop:EW-true} gives $\kaliEXPECTED[W_i]\ge n\cdot g(\alpha)$, and the displayed inequality yields $\mu_i\ge g(\alpha)/2$. In either case,
\[
\mu_i\ \ge\ \min\bigl\{\alpha,\ g(\alpha)/2\bigr\}.
\]
RECE1 is oblivious to $\alpha$, so this inequality holds for every $\alpha\in(0,1)$. Choosing $\alpha=\alpha^\star$ (where $\alpha^\star=g(\alpha^\star)/2$ by definition), we obtain {$\mu_i\ge\alpha^\star > 0.2219$.}
\end{proof}

{To provide a lower bound on $\kaliEXPECTED[v_i(A_i)]$, the proof of Lemma~\ref{lem:expectation} uses (among other things) a lower bound on $\kaliEXPECTED[v_i(\{e_{i,\mathrm{gr}}\})]$, where expectation is taken over the random order of the agents in the greedy phase. However, for the purpose of deriving high probability lower bounds (Lemma~\ref{lem:deviation}), the random variable $v_i(\{e_{i,\mathrm{gr}}\})$ is problematic, because it need not be concentrated around its expectation. For this reason, in subsequent sections, our lower bounds for $v_i(A_i)$ will need to analyze more carefully the effect of those items allocated in the greedy phase, instead of expressing their overall effect as one expectation (which was $\beta$ in the proof of Lemma~\ref{lem:expectation}). The proofs there will hold for every order over the agents in the greedy phase.}

\section{High Probability Lower Bound} \label{sec:high_prob}
\label{subsec:hp-relating-Wprime-Wdblprime}

The high-probability analysis in this section relies on the following assumption, which will be removed in Section~\ref{sec:general-case}.

\begin{assumption}[No item is too large]\label{assm:no-large-item}
For some $\alpha > 0$ (whose value will be determined later), every item $e\in \items$ satisfies $v_i(\{e\})\le \alpha$.
\end{assumption}

The following proposition converts a lower bound on the proxy welfare $W'_i$ (Definition~\ref{def:tilde-v}) into a lower bound on $v_i(A_i)$.

\begin{proposition} \label{prop:W'_i_lower_bound_A_i}
   \label{pro:WversusAlpha}
    For every $\alpha > 0$ and $W'_i$ as in Definition~\ref{def:tilde-v}, if $W'_i \ge (2n-1)\alpha$ and Assumption~\ref{assm:no-large-item} holds, then $v_i(A_i) \ge \alpha$.
\end{proposition}

\begin{proof}
By Claim~\ref{clm:Wi-ge-Wiprime-or-alpha}, either $v_i(A_i) \ge \alpha$ (and we are done) or $W_i \ge W'_i$. In the latter case, since the allocation is EFL, for every agent $j \ne i$ and some item $e \in A_j$, $v_i(A_i) \ge v_i(A_j \setminus \{e\})$. By subadditivity, $v_i(A_j) \le v_i(A_j \setminus \{e\}) + v_i(\{e\}) \le v_i(A_i) + \alpha$, where the last inequality uses Assumption~\ref{assm:no-large-item}. Hence
\[
W_i = v_i(A_i) + \sum_{j \ne i} v_i(A_j) \le n\,v_i(A_i) + (n-1)\alpha.
\]
Combining $W_i \ge W'_i \ge (2n-1)\alpha$ with the upper bound gives $v_i(A_i) \ge \alpha$.
\end{proof}

We now show that, with high probability, the proxy welfare $W'_i$ is not much smaller than $n\cdot g(\alpha)$. Building on the bundle potential $P(R)$ and the drop lemma introduced in Section~\ref{sec:W_i_expected}, we define a single running total $\Lambda_k$ along the random phase that is a submartingale.

\subsection{Filtration and notation}\label{subsec:hp-filtration}

Fix an agent $i$ throughout. Let $G\subseteq\items$ denote the set of items allocated in the greedy phase of RECE1, and index the rounds of the random phase by $k=1,\ldots,T$, where $T:=|\items|-n$ is the number of random-phase rounds. Let $e_k$ denote the item drawn at random-phase round $k$, and let $\mathcal{F}_k$ be the $\sigma$-algebra generated by all randomness {up to and including} round~$k$. For each MMS bundle $B_j$ (with respect to $v'_i$), let $R_j^k$ be the items of $B_j$ still unallocated after round~$k$; in particular $R_j^0=B_j\setminus G_j$ and $R_j^T=\emptyset$. {We continue to use the positive-part notation {$x_+ := \max\{x,\,0\}$} from Section~\ref{sec:W_i_expected}.}

\subsection{The running-total submartingale $\{\Lambda_k\}$}\label{subsec:hp-Mk}

We begin by defining the running-total submartingale.

\begin{definition}[Running total]\label{def:Mk}
For every $k\in\{0,1,\ldots,T\}$,
\[
\Lambda_k\ :=\ \sum_{g\in G} v'_i(\{g\})\ +\ \sum_{\ell=1}^{k} c_i(e_\ell)\ +\ \sum_{j=1}^{n} P(R_j^k).
\]
\end{definition}

\begin{proposition}[Endpoints]\label{prop:Mk-endpoints}
$\Lambda_0\ \ge\ n\cdot g(\alpha)$ and $\Lambda_T\ =\ W'_i$.
\end{proposition}

\begin{proof}
For $\Lambda_0$: iterating the drop lemma (Lemma~\ref{lem:drop}) along the items of $G_j$ gives $P(R_j^0)+\sum_{g\in G_j} v'_i(\{g\})\ge P(B_j)$. Summing over $j\in[n]$ and applying Lemma~\ref{lem:P-bundle-bound},
\[
\Lambda_0\ =\ \sum_{g\in G} v'_i(\{g\})+\sum_{j=1}^n P(R_j^0)\ \ge\ \sum_{j=1}^n P(B_j)\ \ge\ n\cdot g(\alpha).
\]
For $\Lambda_T$: every greedy item $g$ enters an initially empty bundle, so $v'_i(A_g^{<g})=0<\alpha$ and Definition~\ref{def:tilde-v} gives $c_i(g)=v'_i(g\mid\emptyset)=v'_i(\{g\})$; in addition $R_j^T=\emptyset$ and $P(R_j^T)=0$. Therefore
\[
\Lambda_T\ =\ \sum_{g\in G} v'_i(\{g\})+\sum_{\ell=1}^T c_i(e_\ell)+0\ =\ \sum_{g\in G} c_i(g)+\sum_{\ell=1}^T c_i(e_\ell)\ =\ W'_i. \qedhere
\]
\end{proof}

{
\begin{proposition}[Submartingale]\label{prop:Mk-submart}
$\{\Lambda_k\}_{k=0}^{T}$ is a submartingale with respect to $\{\mathcal{F}_k\}$, i.e.\ $\kaliEXPECTED[\Lambda_k-\Lambda_{k-1}\mid\mathcal{F}_{k-1}]\ge 0$ for every $k\in[T]$.
\end{proposition}

{The full proof appears in Appendix~\ref{app:Mk-submart}. The idea is to compare, for each round $k$, the expected credit $c_i(e_k)$ added to $\Lambda_k$ against the expected drop in $\sum_j P(R_j^k)$. Only one bundle's potential changes at round $k$: the one whose remaining set contains $e_k$. So the increment splits as
\[
\Lambda_k-\Lambda_{k-1}\ =\ \underbrace{c_i(e_k)}_{\text{credit added}}\ -\ \underbrace{\bigl[P(R)-P(R\setminus\{e_k\})\bigr]}_{\text{drop in potential}},
\]
where $R$ is the remaining items of the relevant bundle just before round~$k$. Conditioning on $\mathcal{F}_{k-1}$ and on which $B_j$ bundle $e_k$ originates from, the recursion for $P$ (Claim~\ref{obs:P-recursion}) then yields the drop's conditional expectation \emph{exactly}: {$(v'_i(R)-\alpha)_+/|R|$}. The per-round contribution claim (Claim~\ref{clm:per-round-contrib}) yields a matching \emph{lower bound} on the credit's conditional expectation. Subtracting and averaging over all possible $B_j$ gives the submartingale property.}
} 

{
\begin{lemma}[Bounded increments and quadratic variation]\label{lem:Mk-bounded-var}
Deterministically, for every $k\in[T]$:
\begin{enumerate}
\item[(a)] $|\Lambda_k-\Lambda_{k-1}|\ \le\ v'_i(\{e_k\})\ \le\ \alpha$;
\item[(b)] $\sum_{k=1}^{T}(\Lambda_k-\Lambda_{k-1})^2\ \le\ \alpha\,n$.
\end{enumerate}
\end{lemma}

\begin{proof}
Since only one bundle (denoted $R_{j^\star}$) shrinks at round $k$, $\Lambda_k-\Lambda_{k-1}=c_i(e_k)-\bigl[P(R_{j^\star}^{k-1})-P(R_{j^\star}^{k-1}\setminus\{e_k\})\bigr]$.

(a) By Observation~\ref{obs:block-not-opening}, $0\le c_i(e_k)\le v'_i(\{e_k\})$. By Lemma~\ref{lem:drop}, $0\le P(R_{j^\star}^{k-1})-P(R_{j^\star}^{k-1}\setminus\{e_k\})\le v'_i(\{e_k\})$. Both terms in $\Lambda_k-\Lambda_{k-1}$ lie in $[0,v'_i(\{e_k\})]$, so the difference lies in $[-v'_i(\{e_k\}),v'_i(\{e_k\})]$. Finally, $v'_i(\{e_k\})\le v_i(\{e_k\})\le\alpha$ by $v'_i\le v_i$ and Assumption~\ref{assm:no-large-item}.

(b) From (a), $(\Lambda_k-\Lambda_{k-1})^2\le v'_i(\{e_k\})^2\le \alpha\,v'_i(\{e_k\})$. Summing and using subset-additivity of $v'_i$ on each $B_j$,
\[
\sum_{k=1}^{T}(\Lambda_k-\Lambda_{k-1})^2\ \le\ \alpha\sum_{k=1}^{T} v'_i(\{e_k\})\ \le\ \alpha\sum_{e\in\items} v'_i(\{e\})\ =\ \alpha\sum_{j=1}^n v'_i(B_j)\ =\ \alpha\,n,
\]
where the final equality uses $v'_i(B_j)=1$ for every $j$. \qedhere
\end{proof}
}

{\subsection{Direct concentration of $W'_i$}\label{subsec:hp-concentration}

\begin{theorem}[Direct concentration of $W'_i$]\label{thm:Wprime-concentration}
Under Assumption~\ref{assm:no-large-item}, for every $t>0$,
\[
\Prb\bigl[W'_i\ <\ n\cdot g(\alpha)-t\bigr]\ \le\ \exp\left(-\frac{t^2}{2\alpha\,n+2\alpha\,t}\right).
\]
\end{theorem}

\begin{proof}
By Propositions~\ref{prop:Mk-endpoints} and~\ref{prop:Mk-submart}, $\{\Lambda_k\}$ is a submartingale with $\Lambda_0\ge n\cdot g(\alpha)$ deterministically and $\Lambda_T=W'_i$. Set $Z_k:=\Lambda_k-\Lambda_0$; then $\{Z_k\}$ is a submartingale with $Z_0=0$, {increments $D_k:=\Lambda_k-\Lambda_{k-1}\ge -\alpha$ deterministically (Lemma~\ref{lem:Mk-bounded-var}(a))}, and quadratic variation $V_T:=\sum_k D_k^2\le \alpha\,n$ deterministically (Lemma~\ref{lem:Mk-bounded-var}(b)). Theorem~\ref{thm:freedman} applied with $c:=\alpha$ and $v^2:=\alpha\,n$, together with $\Prb[V_T\le \alpha\,n]=1$, gives, for every $t>0$,
\[
\Prb[Z_T\le -t]\ =\ \Prb[Z_T\le -t\text{ and }V_T\le\alpha\,n]\ \le\ \exp\left(-\frac{t^2}{2(\alpha\,n+\alpha\,t)}\right).
\]
Finally, $W'_i=\Lambda_T=\Lambda_0+Z_T\ge n\cdot g(\alpha)+Z_T$, so the event $\{W'_i<n\cdot g(\alpha)-t\}$ is contained in $\{Z_T<-t\}$, yielding the stated bound. \qedhere
\end{proof}

\begin{corollary}[High-probability lower bound on $W'_i$]\label{cor:Wprime-fail-prob}
Under Assumption~\ref{assm:no-large-item}, for every $\eta\in(0,1)$, with probability at least $1-\eta$,
\[
W'_i\ \ge\ n\cdot g(\alpha)\ -\ \delta(\eta)\cdot n,\qquad
\delta(\eta)\ :=\ \frac{\alpha\ln(1/\eta)+\sqrt{\alpha^2\ln^2(1/\eta)+2\alpha\,n\,\ln(1/\eta)}}{n}.
\]
\end{corollary}

\begin{proof}
Set $t=t^\star$ to satisfy $\exp(-t^2/(2\alpha\,n+2\alpha\,t))=\eta$, equivalently $t^2-2\alpha\ln(1/\eta)\,t-2\alpha\,n\,\ln(1/\eta)=0$. The positive root is $t^\star=\alpha\ln(1/\eta)+\sqrt{\alpha^2\ln^2(1/\eta)+2\alpha\,n\,\ln(1/\eta)}=\delta(\eta)\cdot n$. Theorem~\ref{thm:Wprime-concentration} then gives $\Prb[W'_i<n\cdot g(\alpha)-\delta(\eta)\cdot n]\le\eta$.
\end{proof}
}

\subsection{From $W'_i$ to $v_i(A_i)$}\label{subsec:min-lemma}

{Proposition~\ref{thm:alpha-mms-hp} below transfers our bounds on $W'$ to bounds on $v_i(A_i)$.}

{
\begin{proposition}[High-probability $\alpha$-MMS guarantee]\label{thm:alpha-mms-hp}
Assume Assumption~\ref{assm:no-large-item}. Let $(\alpha,n_0)$ be a pair satisfying $\alpha\in(0,1)$ and
\begin{equation}\label{eq:hp-alpha-n}
n\cdot g(\alpha)\ -\ \alpha\ln n\ -\ \sqrt{\alpha^2\ln^2 n+2\alpha\,n\,\ln n}\ \ge\ (2n-1)\,\alpha
\qquad\text{for every }n\ge n_0.
\end{equation}
Then, for every $n\ge n_0$, with probability {at least} $1-1/n$,
\[
v_i(A_i)\ \ge\ \alpha.
\]
\end{proposition}

\begin{proof}
Apply Corollary~\ref{cor:Wprime-fail-prob} with $\eta:=1/n$: with probability at least $1-1/n$,
\[
W'_i\ \ge\ n\cdot g(\alpha)\ -\ \alpha\ln n\ -\ \sqrt{\alpha^2\ln^2 n+2\alpha\,n\,\ln n}\ \ge\ (2n-1)\,\alpha
\]
by~\eqref{eq:hp-alpha-n}. Proposition~\ref{pro:WversusAlpha} on this event gives $v_i(A_i)\ge\alpha$.
\end{proof}
}

\section{Handling Large Items}\label{sec:general-case}

{We now remove Assumption~\ref{assm:no-large-item} and analyze the case where some bundles may contain items whose $v_i$-value exceeds the threshold $\alpha$. The main idea is to separate such bundles and run the high-probability analysis only on the remaining ones.}

\begin{definition}[Heavy items, Light-items bundle and heavy-item bundle]\label{def:light-heavy-bundle}
Fix $\alpha\in(0,1)$ and an agent $i\in\agents$. Call an item \emph{$\alpha$-heavy} for agent $i$ if $v_i(e)\ge\alpha$. Consider the final allocation $A$. A bundle $A_j$ is a \emph{light-items bundle} with respect to agent~$i$ if no item allocated to $A_j$ during the greedy-choice step has $v_i$-value~$\ge\alpha$. Otherwise, $A_j$ is a \emph{heavy-item bundle} with respect to agent~$i$; equivalently, $A_j$ received at least one $\alpha$-heavy item during the greedy phase.
\end{definition}

\begin{definition}[{Light-items bundles}]\label{def:eff-bundles}
Fix $\alpha\in(0,1)$ and an agent $i\in\agents$. For the final allocation $A$,
the set of {\emph{light-items bundles with respect to~$i$}} is

\[
{\allocs^i_{\mathrm{light}}}\ :=\ \Bigl\{\, A_j\ :\ \text{no item allocated to $A_j$ during the greedy-choice step has } v_i\text{-value}\ge \alpha \Bigr\},
\]
{and we write $n_i := |\allocs^i_{\mathrm{light}}|$ for its cardinality.}
\end{definition}

{
\begin{proposition}\label{prop:reduce_to_nprime_i}
{Let {$\allocs^i_{\mathrm{light}}$} and $n_i$ be as in Definition~\ref{def:eff-bundles}.}
{For $\alpha = 0.142$ and $n_i \ge 8$, it holds that $v_i(A_i) \ge \alpha \cdot \mathrm{MMS}(\items,v_i,n)$ with probability greater than $1 - \frac{1}{n_i}$.}
\end{proposition}

\begin{proof}
Let $S$ be the set of all items allocated to bundles outside {$\allocs^i_{\mathrm{light}}$} in the greedy phase. {(All items in $S$ are $\alpha$-heavy for agent $i$.)}

Consider the MMS partition of $\items$ for agent~$i$, and delete every bundle that intersects~$S$.
Since there are $n$ bundles in the MMS partition and each item of~$S$ belongs to at most one of them,
at most $n - n_i$ bundles are deleted; at least $n_i$ remain. Call the collection of all remaining bundles $\bundles^i_{\mathrm{eff}}$.

{Now define the restricted welfare functions. Recall the welfare {functions} from Section~\ref{subsec:proxies}. Let}
\[
\hat{W}_i\ :=\ \sum_{A_j\in {\allocs^i_{\mathrm{light}}}} v_i(A_j),
\qquad
\hat{W}'_i\ :=\ \sum_{A_j\in {\allocs^i_{\mathrm{light}}}} \sum_{e \in A_j} c_i(e).
\]
{Let $\hat{W}''_i$ denote the total contribution of the bundles in $\bundles^i_{\mathrm{eff}}$:}
\[
\hat{W}''_i \ :=\ {\sum_{B_t \in \bundles^i_{\mathrm{eff}}}} C_t,
\]
{where each $C_t$ denotes the contribution of the $t$-th bundle in $\bundles^i_{\mathrm{eff}}$, computed using $c_i(\cdot)$.}

{If $\hat{W}''_i>\hat{W}'_i$, then some bundle outside {$\allocs^i_{\mathrm{light}}$} received an additional item (beyond its {$\alpha$-heavy} item),}
{and Claim~\ref{clm:Wi-ge-Wiprime-or-alpha} gives $v_i(A_i)\ge\alpha$.} Thus, we may assume that $\hat{W}'_i\ge \hat{W}''_i$. {We apply Proposition~\ref{pro:WversusAlpha} with $\hat{W}''_i$ replacing $W'_i$ and $n_i$ replacing $n$. As Assumption~\ref{assm:no-large-item} holds on $\bundles^i_{\mathrm{eff}}$ (by definition), we conclude that if $\hat{W}''_i \ge (2n-1)\alpha$, then $v_i(A_i) \ge \alpha$.

By the expectation lower bound from Section~\ref{sec:W_i_expected}, for every $t$ we have
\[
\mathbb{E}[C_t] \ \ge\ g(\alpha).
\]

{Moreover, the concentration analysis of Section~\ref{sec:high_prob} (Corollary~\ref{cor:Wprime-fail-prob} applied with $n_i$ in place of $n$ and $\eta := 1/n_i$) yields, with probability at least $1-\frac{1}{n_i}$,
\[
\hat{W}''_i \ \ge\ n_i\cdot g(\alpha)\ -\ \alpha\ln n_i\ -\ \sqrt{\alpha^2\ln^2 n_i + 2\alpha\,n_i\,\ln n_i}.
\]
{For $\alpha=0.142$ and every $n_i\ge 8$, this lower bound exceeds $(2n_i-1)\alpha$. To verify this at $n_i=8$, recall (Definition~\ref{def:g-alpha}) $g(\alpha)=1-\alpha-\alpha\ln(1/\alpha)$, so $g(0.142)\approx 0.5808$. The lower bound at $n_i=8$ evaluates to $8\cdot g(0.142)-0.142\ln 8-\sqrt{0.142^2\ln^2 8+2\cdot 0.142\cdot 8\cdot\ln 8}\approx 2.1578$, while $(2\cdot 8-1)\cdot 0.142=2.130$, leaving a slack of $\approx 0.028$. {The slack is monotonically increasing in $n_i$ (rewriting it as $n_i\bigl(g(\alpha)-2\alpha\bigr)+\alpha-\alpha\ln n_i-\sqrt{\alpha^2\ln^2 n_i+2\alpha n_i\ln n_i}$, the linear-in-$n_i$ term has positive slope by Observation~\ref{obs:g-ge-2alpha} and dominates the sublinear correction).}}  Hence Proposition~\ref{pro:WversusAlpha} gives $v_i(A_i)\ge\alpha$, with probability greater than $1-\frac{1}{n_i}$.} 
}
\end{proof}
}

\section{Completing the proof of the main theorem}\label{subsec:constant-all-nprime}

In this section we state and prove the main result (without making use of Assumption~\ref{assm:no-large-item}).

{\begin{definition}[Effective agents]\label{def:eff-agents}
An agent $i$ is \emph{effective} if, at no step of the algorithm, agent~$i$ holds a bundle {that contains an item that} is {$\alpha$-heavy} for $i$. Let $\agents_{\mathrm{eff}}$ denote the set of all effective agents.
\end{definition}}

\begin{claim}\label{clm_eff_agents}
{For every agent $i$, either $i\in\agents_{\mathrm{eff}}$ or $v_i(A_i)\ge\alpha\cdot MMS(\items,v_i,n)$.}
\end{claim}

\begin{proof}
{If $i\notin\agents_{\mathrm{eff}}$, then at some step agent~$i$ held a bundle of $v_i$-value {at least} $\alpha$. Since the Cycle Rotation Process only increases each agent's value (Claim~\ref{clm:envy_cycle_sink}) and items are never removed, $v_i(A_i)\ge\alpha\cdot MMS(\items,v_i,n)$.}
\end{proof}

{
\begin{theorem}\label{thm:constant_bound_all_nprime}
{Let {$\alpha=0.142$}. Then, with positive probability, \textsc{RECE1} outputs an allocation that satisfies

\[
v_i(A_i)\ \ge\ {0.142}\cdot MMS(\items,v_i,n)
\]
for every agent $i$.}
\end{theorem}

\begin{proof}
Fix $\alpha=0.142$.

{If $i\notin\agents_{\mathrm{eff}}$, then $v_i(A_i)\ge\alpha\cdot MMS(\items,v_i,n)$ by Claim~\ref{clm_eff_agents}, and we are done.}

Assume henceforth that $i\in\agents_{\mathrm{eff}}$.
{For every effective agent~$i$, {recall $n_i$ from Definition~\ref{def:eff-bundles},} and let $u := |\agents_{\mathrm{eff}}|$.}
Relabel the effective agents so that $n_1 \ge n_2 \ge \cdots \ge n_u$ (breaking ties arbitrarily).
Define
\[
\ell\ :=\ \max\bigl\{\,k\in[u] : k \le n_k\bigr\}.
\]

For each effective agent $i$ with index at most $\ell$:
\begin{itemize}
\item {If $n_i \ge 8$: by Proposition~\ref{prop:reduce_to_nprime_i},
\[
v_i(A_i)\ \ge\ \alpha\cdot MMS(\items,v_i,n)
\]
with probability strictly greater than $1-\frac{1}{n_i}$.}
\item {If $n_i < 8$: every heavy-item bundle $A_j$ (w.r.t.~$i$) received an $\alpha$-heavy item~$e$ during the greedy phase. If $|A_j|\ge 2$, then $A_j$ received an additional item beyond~$e$, and Claim~\ref{clm:Wi-ge-Wiprime-or-alpha} gives $v_i(A_i)\ge\alpha\cdot MMS(\items,v_i,n)$. Otherwise, every heavy-item bundle contains exactly one item (its $\alpha$-heavy item). Let $S^i$ be the set of these items; then $|S^i|=n-n_i$ and each item in $S^i$ is $\alpha$-heavy for~$i$. The allocation restricted to the $n_i$ light-items bundles (i.e.\ items $\items\setminus S^i$) is EFL for the agents holding these bundles. By Proposition~\ref{prop:EF1_is1_over_n}, $v_i(A_i)\ge\frac{1}{n_i}\cdot MMS(\items\setminus S^i,v_i,n_i)$. By Claim~\ref{clm:not_too_small_mms}, $MMS(\items\setminus S^i,v_i,n_i)\ge MMS(\items,v_i,n)$. Since $n_i\le 7$, we have $\frac{1}{n_i}\ge\frac{1}{7}>0.142=\alpha$, so $v_i(A_i)\ge\alpha\cdot MMS(\items,v_i,n)$ deterministically.}
\end{itemize}

Applying a union bound over the {the first $\ell$ effective} agents:
since each failure probability is strictly less than {$\frac{1}{n_\ell}$}, the total failure probability is strictly less than~$1$.
Thus, with positive probability, every agent with index $\le \ell$ satisfies $v_i(A_i)\ge \alpha\cdot MMS(\items,v_i,n)$.

\paragraph{{Effective} agents with index greater than $\ell$.}
Condition on the above event, {that} every agent $k$ with index at most $\ell$ satisfies $v_k(A_k)\ge\alpha$.
{Since $k\in\agents_{\mathrm{eff}}$, agent~$k$ never held a single-item bundle of $v_k$-value $\ge\alpha$. Combined with $v_k(A_k)\ge\alpha$, this forces $|A_k|\ge 2$.}

{We prove} that every {effective} agent with index greater than $\ell$ also achieves value~$\ge\alpha$.
Fix such an agent; call her~$j$.
Since the $n_i$ are sorted in non-increasing order and $j > \ell$,
we have $n_j \le n_{\ell+1} \le \ell$, so $n_j \le \ell$. {As} {$A_j\in\allocs^j_{\mathrm{light}}$}, {bundle} $A_j$ {is} one of the $n_j$ {light-items bundles}. So, at most $n_j-1$ of the first~$\ell$ agents of~$\agents^j_{\mathrm{eff}}$ have their bundle in~{$\allocs^j_{\mathrm{light}}$}.
Since $n_j - 1 < \ell$, there exists a first-$\ell$ {agent $p$,} s.t.\ {$A_p\notin\allocs^j_{\mathrm{light}}$}.

{By definition of {$\allocs^j_{\mathrm{light}}$},} agent~$p$'s bundle~$A_p$ contains an item~$e$ with $v_j(e)\ge\alpha$.

{Since $|A_p| \ge 2$, bundle~$A_p$ received an additional item in the random phase. By the algorithm's construction, this can only happen when agent~$p$ is a source in the envy graph, meaning no agent envies~$A_p$. In particular $v_j(A_j) \ge v_j(A_p) \ge v_j(e) \ge \alpha \cdot MMS(\items, v_j, n)$.} 

Combining both groups, with positive probability every effective agent achieves the claimed bound.
\end{proof}
}

\ifshowendacks
{
\subsection*{Acknowledgements}

\ifanonymousmode\else This research was supported in part by the Israel Science Foundation (grant No. 1122/22). \fi{The preparation of this version of the paper was assisted by insightful comments and suggestions that AI tools {(GPT-5.5, Opus 4.7)} provided on a previous version of this paper.}
}
\fi

\newcommand{\renderappendices}{\appendix
  \section{Missing proofs from Section~\ref{sec:RECE1}}\label{app:missing-ece}

{
\subsection{ECE outputs an EF1 allocation}\label{app:ECE_EF1}
}

{Here we prove Proposition~\ref{prp:ECE_EF1}.}

\begin{proof}
Since a new item $e$ is always allocated to an agent that no other agent envies, if $\hat{A}$ is EF1 just before $e$ is allocated, 
then after allocating $e$  to such an agent the updated partial allocation remains EF1, by definition.
It remains to show that EF1 is preserved by the Cycle Rotation Process. 
{This holds because cycle rotations only permute existing bundles among agents (no bundle's contents change), and each agent's value for her own bundle does not decrease.}

\end{proof}

\subsection{The allocation produced by RECE1 is EFL}\label{app:RECE1_EFL_as}

{Here we prove Proposition~\ref{prop:RECE1_EFL_as}.}

\begin{proof}[Proof sketch] 
Since RECE1 is an instantiation of ECE, Proposition~\ref{prp:ECE_EF1} gives EF1.
For the stronger EFL property, fix agents $i,j$ with $|A_j|\ge 2$ and let $g$ be the last item added to $A_j$ (necessarily in the random phase).
At the round $g$ was allocated, the recipient had no incoming edges in the envy graph, so removing $g$ eliminates $i$'s envy: $v_i(A_i) \ge v_i(A_j \setminus \{g\})$.
Moreover, $g$ was still unallocated during the greedy phase, so by the greedy choice $v_i(g) \le v_i(e_i) \le v_i(A_i)$.
Together, these two conditions establish EFL.
\end{proof}

\section{Missing proofs from Section~\ref{sec:W_i_expected}}\label{app:Expected_lower_bound}

\subsection{Proof of Claim \ref{clm:marginal-beta}} \label{app:marginal-beta}
This is a well known fact. For completeness, we present a short proof.
\begin{proof}
By monotonicity, $v(A\cup B)\ge v(B)$, and therefore
\[
v(B\mid A)=v(A\cup B)-v(A)\ \ge\ v(B)-v(A)\ \ge\ \beta-\alpha.
\]
Consider adding the items of $B=(e_1,\dots,e_{|B|})$ one by one to $A$.
Submodularity implies:
\[
v\left(B\mid A\right)
=\sum_{l=1}^{\left|B\right|}v\left(e_l\mid A\cup_{k<l}e_{k}\right)
\leq\sum_{l=1}^{\left|B\right|}v\left(e_l\mid A\right).
\]
Combining the two inequalities yields
\[
(\beta-\alpha)\ \le\ \sum_{l=1}^{|B|} v(e_l\mid A).
\]
Finally, the random item $e$ is uniform in $B$, so averaging over $|B|$ gives the claimed inequality.
\end{proof}

\subsection{Averaging bound for monotone submodular valuations}\label{app:expected_lower_sub}

{We prove the standard averaging bound for monotone submodular functions: if $S$ is a uniformly random $r$-subset of a set $B$ with $|B|=k$, then $\mathbb{E}[v(S)] \ge (r/k)\,v(B)$. }
{
\begin{proof}
By Lemma~\ref{lem:MMSoverAdditive}, there exists an additive $v' \le v$ with $v'(B) = v(B)$.  Since each $e \in B$ belongs to a uniformly random $r$-subset with probability $r/k$, linearity of expectation gives
\[
\mathbb{E}[v(S)] \;\ge\; \mathbb{E}[v'(S)] \;=\; \sum_{e \in B} v'(\{e\})\,\Pr[e \in S] \;=\; \frac{r}{k}\,v'(B) \;=\; \frac{r}{k}\,v(B). \qedhere
\]
\end{proof}}

\subsection{Calculation omitted from Proposition \ref{prop:expected_c_j}} \label{app:expected_c_j}

\begin{lemma}[Sum bound for $g(\alpha)$]\label{lem:g-sum-bound}
For every integer $k\ge 1$ and every $\alpha\in(0,1)$,
\[
\sum_{r=1}^{k}\Bigl(\frac{1}{k}-\frac{\alpha}{r}\Bigr)_+\ \ge\ g(\alpha).
\]
\end{lemma}

{
\begin{proof}
If $k\alpha<1$, then $\alpha/r\le\alpha<1/k$ for every $r\in\{1,\dots,k\}$, so all $k$ summands are positive and the LHS equals $1-\alpha H_k$, where $H_k:=\sum_{r=1}^{k}1/r$. Since $H_k\le 1+\ln k$ and $k<1/\alpha$,
\[
1-\alpha H_k\ \ge\ 1-\alpha(1+\ln k)\ \ge\ 1-\alpha-\alpha\ln(1/\alpha)\ =\ g(\alpha).
\]

Assume henceforth that $k\alpha\ge 1$. Let $a:=\lfloor \alpha k\rfloor\ge 1$, and write $x:=\alpha k$, so $x\in[a,a+1)$. The positive summands are exactly those with $r>a$, hence
\[
\sum_{r=1}^{k}\Bigl(\frac{1}{k}-\frac{\alpha}{r}\Bigr)_+\ =\ \frac{k-a}{k}-\alpha\sum_{r=a+1}^{k}\frac{1}{r}.
\]
Use the elementary bound
\[
\sum_{r=a+1}^{k}\frac{1}{r}\ \le\ \frac{1}{a+1}+\int_{a+1}^{k}\frac{dt}{t}\ =\ \frac{1}{a+1}+\ln\frac{k}{a+1}.
\]
Therefore
\[
\sum_{r=1}^{k}\Bigl(\frac{1}{k}-\frac{\alpha}{r}\Bigr)_+\ \ge\ 1-\frac{a}{k}-\frac{\alpha}{a+1}-\alpha\ln\frac{k}{a+1}.
\]
Subtracting $g(\alpha)=1-\alpha-\alpha\ln(1/\alpha)$ and using $1/\alpha=k/x$ (so $\ln(1/\alpha)=\ln(k/x)$), the remaining gap, divided by $\alpha$, equals
\[
\psi(x)\ :=\ 1-\frac{a}{x}-\frac{1}{a+1}-\ln\frac{x}{a+1}.
\]
We wish to show that \(\psi(x)\ge 0\) for every possible \(x\in[a,a+1)\).
Fix \(a\). Differentiating
with respect to \(x\) gives
\[
\psi'(x)\ =\ \frac{a}{x^2}-\frac{1}{x}\ =\ \frac{a-x}{x^2}\ \le\ 0,
\]
so $\psi$ is non-increasing on $[a,a+1)$, and
\[
\psi(x)\ \ge\ \lim_{y\uparrow a+1}\psi(y)\ =\ 1-\frac{a}{a+1}-\frac{1}{a+1}-\ln 1\ =\ 0.
\]
This proves the bound.
\end{proof}
}

{\subsection{Proof of the drop lemma (Lemma~\ref{lem:drop})}\label{app:drop-lemma}}

{
\begin{proof}
For every $j\in[n]$, every nonempty $R\subseteq B_j$, and every $e\in R$, we prove the two inequalities of the lemma:
\[
P(R)-v'_i(\{e\})\ \le\ P(R\setminus\{e\})\ \le\ P(R).
\]
Throughout, the function $(x-\alpha)_+=\max\{x-\alpha,0\}$ is monotone non-decreasing and $1$-Lipschitz on $\mathbb{R}$, so $0\le (x+a-\alpha)_+-(x-\alpha)_+\le a$ for every $x\ge 0$ and $a\ge 0$. We also use subset-additivity of $v'_i$ on $B_j$: $v'_i(T)=v'_i(T\setminus\{e\})+v'_i(\{e\})$ for every $T\subseteq B_j$ and every $e\in T$.

Fix $R$ and $e$ as above. Set $s:=|R|$, $a:=v'_i(\{e\})$, $R^-:=R\setminus\{e\}$, and abbreviate $D(R,e):=P(R)-P(R^-)$; the lemma's two inequalities are then equivalent to the single bound $D(R,e)\in[0,a]$, which we prove by induction on $s$.

\paragraph{Step 1: A recursion for $D(R,e)$.}
By Claim~\ref{obs:P-recursion} applied to $R$, splitting the inner expectation $\kaliEXPECTED_{e'\sim\mathbb{U}(R)}$ on whether $e'=e$ or $e'\in R^-$, and using $v'_i(R)=v'_i(R^-)+a$,
\begin{equation}\label{eq:drop-PR}
P(R)\ =\ \frac{(v'_i(R^-)+a-\alpha)_+}{s}\ +\ \frac{1}{s}\,P(R^-)\ +\ \frac{s-1}{s}\,\kaliEXPECTED_{e'\sim\mathbb{U}(R^-)}\bigl[P(R\setminus\{e'\})\bigr].
\end{equation}
For $s\ge 2$, applying Claim~\ref{obs:P-recursion} to the nonempty set $R^-$,
\begin{equation}\label{eq:drop-PRminus}
P(R^-)\ =\ \frac{(v'_i(R^-)-\alpha)_+}{s-1}\ +\ \kaliEXPECTED_{e'\sim\mathbb{U}(R^-)}\bigl[P(R^-\setminus\{e'\})\bigr].
\end{equation}
For $e'\in R^-$, $R\setminus\{e'\}\ni e$ and $R^-\setminus\{e'\}=(R\setminus\{e'\})\setminus\{e\}$, so $P(R\setminus\{e'\})-P(R^-\setminus\{e'\})=D(R\setminus\{e'\},e)$. From~\eqref{eq:drop-PR},
\[
D(R,e)\ =\ P(R)-P(R^-)\ =\ \frac{(v'_i(R^-)+a-\alpha)_+}{s}\ -\ \frac{s-1}{s}P(R^-)\ +\ \frac{s-1}{s}\,\kaliEXPECTED_{e'\sim\mathbb{U}(R^-)}\bigl[P(R\setminus\{e'\})\bigr].
\]
Substituting~\eqref{eq:drop-PRminus} for $P(R^-)$ and simplifying,
\begin{align*}
D(R,e)
&=\ \frac{(v'_i(R^-)+a-\alpha)_+}{s}\ -\ \frac{s-1}{s}\!\left[\frac{(v'_i(R^-)-\alpha)_+}{s-1}+\kaliEXPECTED_{e'}[P(R^-\setminus\{e'\})]\right]\\
&\quad +\ \frac{s-1}{s}\,\kaliEXPECTED_{e'}[P(R\setminus\{e'\})]\\
&=\ \frac{(v'_i(R^-)+a-\alpha)_+-(v'_i(R^-)-\alpha)_+}{s}\ +\ \frac{s-1}{s}\,\kaliEXPECTED_{e'\sim\mathbb{U}(R^-)}\bigl[D(R\setminus\{e'\},e)\bigr].
\end{align*}
Hence, for $s\ge 2$,
\begin{equation}\label{eq:drop-recur}
D(R,e)\ =\ \frac{(v'_i(R^-)+a-\alpha)_+-(v'_i(R^-)-\alpha)_+}{s}\ +\ \frac{s-1}{s}\,\kaliEXPECTED_{e'\sim\mathbb{U}(R^-)}\bigl[D(R\setminus\{e'\},e)\bigr].
\end{equation}

\paragraph{Step 2: Induction on $|R|$.}
\emph{Base ($s=1$, i.e., $R=\{e\}$).} By Claim~\ref{obs:P-recursion} applied to $\{e\}$ together with $P(\emptyset)=0$, $P(\{e\})=(a-\alpha)_+$. Hence $D(\{e\},e)=(a-\alpha)_+-0\in[0,a]$. 

\emph{Inductive step ($s\ge 2$).} Assume $D(R',e')\in[0,v'_i(\{e'\})]$ for every (subset, element) pair with subset of size $<s$. The first term of~\eqref{eq:drop-recur} is $(v'_i(R^-)+a-\alpha)_+-(v'_i(R^-)-\alpha)_+\in[0,a]$ (with $x:=v'_i(R^-)\ge 0$). The second term, $\kaliEXPECTED_{e'\sim\mathbb{U}(R^-)}[D(R\setminus\{e'\},e)]$, is an average of quantities each in $[0,a]$ by the inductive hypothesis (each $R\setminus\{e'\}$ contains $e$ and has size $s-1<s$). A convex combination, with weights $1/s$ and $(s-1)/s$, of values in $[0,a]$ lies in $[0,a]$, so $D(R,e)\in[0,a]$.

In either case $D(R,e)=P(R)-P(R^-)\in[0,a]$, equivalently $P(R)-v'_i(\{e\})\le P(R\setminus\{e\})\le P(R)$.
\end{proof}
}

\subsection{Proof of Proposition~\ref{prop:EW-true}}\label{app:EW-true}

\begin{proof}
Set $d_i(e):=v_i(e\mid A_e^{<e})$, so $W_i=\sum_{e\in\items}d_i(e)$. Fix an MMS bundle $B_j$, partition $B_j=G_j\sqcup R_j^0$ as in Proposition~\ref{prop:expected_c_j}, set $k:=|B_j|$ and $k':=|R_j^0|$.

\paragraph{Per-round bound.} The recipient bundle $A_{e_r}^{<e_r}$ is an envy-graph sink; combined with monotonicity of $v_i(A_i)$ along the run, $X_r:=v_i(A_{e_r}^{<e_r})\le v_i(A_i)$ deterministically, so $\kaliEXPECTED[X_r]\le\mu_i\le\alpha$. Applying Claim~\ref{clm:marginal-beta} to $v_i$ with $A:=A_{e_r}^{<e_r}$ and source set $R$ (a uniform $r$-subset of $R_j^0$ at round $r$), then the inequality $\kaliEXPECTED[\max\{Y,0\}]\ge \max\{\kaliEXPECTED[Y],0\}$ applied with $Y:=v_i(R)-\alpha$ and $\kaliEXPECTED[v_i(R)]\ge \kaliEXPECTED[v'_i(R)]=(r/k')\,v'_i(R_j^0)$ (Claim~\ref{clm:avg-submodular}, additivity of $v'_i$ on $B_j$, and $v'_i\le v_i$), gives
\begin{equation}\label{eq:EWtrue-perround}
\kaliEXPECTED[d_i(e_r)]\ \ge\ \Bigl(\tfrac{v'_i(R_j^0)}{k'}-\tfrac{\alpha}{r}\Bigr)_+.
\end{equation}

\paragraph{Reduction to the $W'_i$ argument.} From here the proof   {has a structure similar to that of the proof} of Proposition~\ref{prop:expected_c_j}, with $d_i$ in place of $c_i$, and~\eqref{eq:EWtrue-perround} in place of Claim~\ref{clm:per-round-contrib}. 

{Set $\tau:=v'_i(R_j^0)$ Since $v'_i$ is additive on $B_j$ with $v'_i(B_j)=1$ and $R_j^0\subseteq B_j$, $\tau\in[0,1]$. The greedy items of $B_j$ contribute
\[
\sum_{g\in G_j}d_i(g)\ =\ \sum_{g\in G_j}v_i(\{g\})\ \ge\ \sum_{g\in G_j}v'_i(\{g\})\ =\ 1-\tau,
\]
using $v'_i\le v_i$ and additivity of $v'_i$ on $B_j$. We split into two cases.

\textit{Case 1: $\tau\le\alpha$.} The random-phase contribution is non-negative, so the greedy contribution alone gives
\[
\kaliEXPECTED\Bigl[\sum_{e\in B_j}d_i(e)\Bigr]\ \ge\ 1-\tau\ \ge\ 1-\alpha\ \ge\ g(\alpha),
\]
where the last inequality, {which can be written as} {$1-\alpha-g(\alpha)=\alpha\ln(1/\alpha)\ge 0$}, {holds} for $\alpha\in(0,1)$.

\textit{Case 2: $\tau>\alpha$.} Set $\theta:=\alpha/\tau\in(0,1)$. Summing~\eqref{eq:EWtrue-perround} over $r=1,\dots,k'$ and applying Lemma~\ref{lem:g-sum-bound} with parameter $\theta$,
\[
\kaliEXPECTED\Bigl[\sum_{e\in R_j^0}d_i(e)\Bigr]\ \ge\ \tau\sum_{r=1}^{k'}\Bigl(\tfrac{1}{k'}-\tfrac{\theta}{r}\Bigr)_+\ \ge\ \tau\,g(\theta)\ =\ \tau\,g(\alpha/\tau).
\]
Adding the greedy contribution,
\[
\kaliEXPECTED\Bigl[\sum_{e\in B_j}d_i(e)\Bigr]\ \ge\ (1-\tau)+\tau\,g(\alpha/\tau)\ =\ {1-\alpha-\alpha\ln(\tau/\alpha)}\ \ge\ {1-\alpha-\alpha\ln(1/\alpha)}\ =\ g(\alpha),
\]
using $\tau\le 1$ in the penultimate inequality.

In either case, $\kaliEXPECTED[\sum_{e\in B_j}d_i(e)]\ge g(\alpha)$. Summing over $j\in[n]$ yields $\kaliEXPECTED[W_i]\ge n\cdot g(\alpha)$.}\qedhere
\end{proof}

\section{Missing proofs from Section~\ref{sec:high_prob}}\label{app:missing-high_prob}

{\subsection{Proof of Proposition~\ref{prop:Mk-submart}}\label{app:Mk-submart}

\textbf{Proposition~\ref{prop:Mk-submart}} \textbf{.} \emph{$\{\Lambda_k\}_{k=0}^{T}$ is a submartingale with respect to $\{\mathcal{F}_k\}$, i.e.\ $\kaliEXPECTED[\Lambda_k-\Lambda_{k-1}\mid\mathcal{F}_{k-1}]\ge 0$ for every $k\in[T]$.}

\begin{proof}
Fix $k\in[T]$ and let $j^\star\in[n]$ be the (random) index for which $e_k\in B_{j^\star}$. It suffices to show
\begin{equation}\label{eq:Mk-conditional-goal}
\kaliEXPECTED[\Lambda_k-\Lambda_{k-1}\mid\mathcal{F}_{k-1},\,j^\star=j]\ \ge\ 0
\end{equation}
for every $j$, since averaging over $j^\star$ then yields the unconditional submartingale property.

\paragraph{Increment decomposition.} Only $R_j^{k-1}$ shrinks at round $k$; the other remaining sets are unchanged. Hence the only term in $\Lambda_k$ that changes besides the credit $c_i(e_k)$ is the bundle potential of $B_j$, giving
\begin{equation}\label{eq:Mk-increment}
\Lambda_k-\Lambda_{k-1}\ =\ c_i(e_k)\ -\ \Delta P_k,
\qquad
\Delta P_k\ :=\ P(R_j^{k-1})-P(R_j^{k-1}\setminus\{e_k\}).
\end{equation}
Goal~\eqref{eq:Mk-conditional-goal} thus reduces to
\begin{equation}\label{eq:Mk-credit-vs-drop}
\kaliEXPECTED[c_i(e_k)\mid\mathcal{F}_{k-1},\,j^\star=j]\ \ge\ \kaliEXPECTED[\Delta P_k\mid\mathcal{F}_{k-1},\,j^\star=j].
\end{equation}

\paragraph{Drop side (RHS of~\eqref{eq:Mk-credit-vs-drop}).} Conditional on $\mathcal{F}_{k-1}$ and $\{j^\star=j\}$, the item $e_k$ is uniformly distributed on $R_j^{k-1}$. Therefore
\[
\kaliEXPECTED[\Delta P_k\mid\mathcal{F}_{k-1},\,j^\star=j]\ =\ \kaliEXPECTED_{e'\sim\mathbb{U}(R_j^{k-1})}\bigl[P(R_j^{k-1})-P(R_j^{k-1}\setminus\{e'\})\bigr]\ =\ {\frac{(v'_i(R_j^{k-1})-\alpha)_+}{|R_j^{k-1}|}},
\]
where the last equality is Claim~\ref{obs:P-recursion} (the recursion for $P$, rearranged).

\paragraph{Credit side (LHS of~\eqref{eq:Mk-credit-vs-drop}).} Let $A:=A_{e_k}^{<e_k}$ denote the contents of the recipient bundle just before $e_k$ is allocated.

\emph{Case 1: $v'_i(A)<\alpha$.} Then $c_i(e_k)=v'_i(e_k\mid A)$. Apply Claim~\ref{clm:marginal-beta} to $v'_i$ with the conditioning set $A$ (which satisfies $v'_i(A)<\alpha$) and source set $B:=R_j^{k-1}\subseteq B_j$, using $\beta:=v'_i(R_j^{k-1})$. Together with the uniformity of $e_k$ on $R_j^{k-1}$,
\[
\kaliEXPECTED[c_i(e_k)\mid\mathcal{F}_{k-1},\,j^\star=j]\ =\ \kaliEXPECTED_{e\sim\mathbb{U}(R_j^{k-1})}[v'_i(e\mid A)]\ \ge\ \frac{(v'_i(R_j^{k-1})-\alpha)_+}{|R_j^{k-1}|}.
\]

\emph{Case 2: $v'_i(A)\ge\alpha$.} Then $c_i(e_k)=v'_i(\{e_k\})$. By the uniformity of $e_k$ on $R_j^{k-1}$ and subset-additivity of $v'_i$ on $B_j$,
\[
\kaliEXPECTED[c_i(e_k)\mid\mathcal{F}_{k-1},\,j^\star=j]\ =\ \frac{1}{|R_j^{k-1}|}\sum_{e\in R_j^{k-1}} v'_i(\{e\})\ =\ \frac{v'_i(R_j^{k-1})}{|R_j^{k-1}|}\ {\ge\ \frac{(v'_i(R_j^{k-1})-\alpha)_+}{|R_j^{k-1}|}},
\]
{since $(x-\alpha)_+ \le x$ for every $x\ge 0$.}

In both cases, the LHS of~\eqref{eq:Mk-credit-vs-drop} is at least {$(v'_i(R_j^{k-1})-\alpha)_+/|R_j^{k-1}|$}, which equals the RHS. This proves~\eqref{eq:Mk-credit-vs-drop}, hence~\eqref{eq:Mk-conditional-goal}. \qedhere
\end{proof}

}

\subsection{Proof of Lemma \ref{lem:MMSoverAdditive}}\label{app:vprime}

{In our proof of Lemma \ref{lem:MMSoverAdditive} we make use of the following claim.}

\begin{claim}
\label{cl:decreaseSubmodular}
Let $v$ be a monotone submodular valuation function and consider any item $e\in\items$.
Then, for every $0 \le t {\le} v(\{e\})$, the following function $v'$ is a monotone submodular valuation function:
\begin{enumerate}[label=\textbf{\arabic*.}, leftmargin=2.2em]
    \item $v'(\{e\}) = t$.
    \item $v'(S) = v(S)$ for every $S \subseteq (\items \setminus \{e\})$.
    \item $v'(S \cup \{e\}) = \min\{\,v(S \cup \{e\}),\ v'(S) + t\,\}$ for every $S \subseteq (\items \setminus \{e\})$.
\end{enumerate}
Moreover, $v'(S)\le v(S)$ for every $S\subseteq \items$.
\end{claim}

\begin{proof}
\textbf{Pointwise domination.}
If $e\notin S$, then $v'(S)=v(S)$. If $e\in S$, write $S=T\cup\{e\}$ with $T\subseteq \items\setminus\{e\}$.
Then
\[
v'(S)=v'(T\cup\{e\})=\min\{v(T\cup\{e\}),\,v'(T)+t\}\le v(T\cup\{e\})=v(S).
\]

\medskip
\noindent\textbf{Monotonicity.}
Let $S\subseteq \items$ and $x\in \items\setminus S$. We prove $v'(S\cup\{x\})\ge v'(S)$.

If $e\notin S$ and $x\neq e$, then $v'(S)=v(S)$ and $v'(S\cup\{x\})=v(S\cup\{x\})\ge v(S)=v'(S)$.

If $x=e$ and $e\notin S$, then
\[
v'(S\cup\{e\})=\min\{v(S\cup\{e\}),\,v(S)+t\}\ge v(S),
\]
since both terms inside the minimum are at least $v'(S)$.

If $e\in S$ and $x\neq e$, write $S=T\cup\{e\}$ with $T\subseteq \items\setminus\{e\}$.
Then $S\cup\{x\}=(T\cup\{x\})\cup\{e\}$, and using monotonicity of $v$ and Item~2,
\[
v(T\cup\{x\}\cup\{e\})\ge v(T\cup\{e\}),\qquad v'(T\cup\{x\})=v(T\cup\{x\})\ge v(T)=v'(T).
\]
Therefore
\[
v'(S\cup\{x\})
=\min\{v(T\cup\{x\}\cup\{e\}),\,v'(T\cup\{x\})+t\}
\ge \min\{v(T\cup\{e\}),\,v'(T)+t\}=v'(S).
\]

\medskip
\noindent\textbf{Submodularity.}
We prove $v'(S)+v'(T)\ge v'(S\cup T)+v'(S\cap T)$ for all $S,T\subseteq\items$.

If $e\notin S\cup T$, then $v'=v$ on all four sets and the inequality holds since $v$ is submodular.

Assume $e\in S\cup T$. We handle the remaining cases.

\smallskip
\noindent\emph{Case 1: $e\in S$ and $e\notin T$.}
Write $S=S_0\cup\{e\}$ where $S_0\subseteq \items\setminus\{e\}$, and note $T\subseteq \items\setminus\{e\}$.
Then $v'(T)=v(T)$ and $v'(S\cap T)=v(S_0\cap T)$.
Also,
\[
v'(S)=\min\{v(S_0\cup\{e\}),\,v(S_0)+t\},\quad
v'(S\cup T)=\min\{v(S_0\cup T\cup\{e\}),\,v(S_0\cup T)+t\}.
\]
If $v'(S)=v(S_0\cup\{e\})$, then using $v'(S\cup T)=v'(S_0\cup T\cup\{e\})\le v(S_0\cup T\cup\{e\})$,
\[
v'(S)+v'(T)=v(S_0\cup\{e\})+v(T)\ge v(S_0\cup T\cup\{e\})+v(S_0\cap T)\ge v'(S\cup T)+v'(S\cap T),
\]
by submodularity of $v$.
If $v'(S)=v(S_0)+t$, then using $v'(S\cup T)\le v(S_0\cup T)+t$,
\[
v'(S)+v'(T)=v(S_0)+t+v(T)\ge v(S_0\cup T)+v(S_0\cap T)+t\ge v'(S\cup T)+v'(S\cap T),
\]
again by submodularity of $v$.

\smallskip
\noindent\emph{Case 2: $e\in S$ and $e\in T$.}
Write $S=S_0\cup\{e\}$ and $T=T_0\cup\{e\}$ where $S_0,T_0\subseteq \items\setminus\{e\}$.
By the definition of $v'$, we have
\[
v'(S)=\min\{\,v(S),\,v(S_0)+t\,\},\qquad
v'(T)=\min\{\,v(T),\,v(T_0)+t\,\}.
\]
Also,
\[
S\cup T=(S_0\cup T_0)\cup\{e\},\qquad S\cap T=(S_0\cap T_0)\cup\{e\},
\]
and therefore
\begin{align*}
v'(S\cup T)
&=\min\{\,v(S\cup T),\,v(S_0\cup T_0)+t\,\},\\
v'(S\cap T)
&=\min\{\,v(S\cap T),\,v(S_0\cap T_0)+t\,\}.
\end{align*}

We prove $v'(S)+v'(T)\ge v'(S\cup T)+v'(S\cap T)$ by a case analysis according to which term
attains each minimum.

\smallskip
\noindent\underline{Subcase 2.1:} $v'(S)=v(S)$ and $v'(T)=v(T)$.
Then $v'(S\cup T)\le v(S\cup T)$ and $v'(S\cap T)\le v(S\cap T)$, hence
\[
v'(S)+v'(T)=v(S)+v(T)\ge v(S\cup T)+v(S\cap T)\ge v'(S\cup T)+v'(S\cap T),
\]
where the first inequality is submodularity of $v$.

\smallskip
\noindent\underline{Subcase 2.2:} $v'(S)=v(S)$ and $v'(T)=v(T_0)+t$.
Since $T=T_0\cup\{e\}$, we also have $S\cup T=S\cup T_0$ and $S\cap T=(S\cap T_0)\cup\{e\}$.
Moreover, $v'(S\cup T)\le v(S\cup T)$ and $v'(S\cap T)\le v(S\cap T_0)+t$.
Therefore,
\begin{align*}
v'(S)+v'(T)
&=v(S)+v(T_0)+t\\
&\ge v(S\cup T_0)+v(S\cap T_0)+t \\
&= v(S\cup T)+\bigl(v(S\cap T_0)+t\bigr)\\
&\ge v'(S\cup T)+v'(S\cap T),
\end{align*}
where the inequality is submodularity of $v$ applied to the pair $(S,T_0)$.

\smallskip
\noindent\underline{Subcase 2.3:} $v'(S)=v(S_0)+t$ and $v'(T)=v(T)$.
This is symmetric to Subcase 2.2 (swap the roles of $S$ and $T$).

\smallskip
\noindent\underline{Subcase 2.4:} $v'(S)=v(S_0)+t$ and $v'(T)=v(T_0)+t$.
In this subcase,
\[
v'(S)+v'(T)=v(S_0)+v(T_0)+2t.
\]
Also, $v'(S\cup T)\le v(S_0\cup T_0)+t$ and $v'(S\cap T)\le v(S_0\cap T_0)+t$.
Thus,
\begin{align*}
v'(S\cup T)+v'(S\cap T)
&\le \bigl(v(S_0\cup T_0)+t\bigr)+\bigl(v(S_0\cap T_0)+t\bigr)\\
&= v(S_0\cup T_0)+v(S_0\cap T_0)+2t\\
&\le v(S_0)+v(T_0)+2t\\
&= v'(S)+v'(T),
\end{align*}
where the second inequality is submodularity of $v$ applied to $(S_0,T_0)$.

This completes the case analysis for $e\in S$ and $e\in T$, and therefore establishes submodularity of $v'$.

Thus $v'$ is submodular in all cases, completing the proof.
\end{proof}

\subsubsection{Additivizing disjoint bundles while preserving submodularity}

\paragraph{Lemma \ref{lem:MMSoverAdditive}}\label{par:MMSoverAdditive}
Let $v$ be a monotone submodular valuation function and let $B_1,\ldots,B_n$ be disjoint subsets of $\items$.
Then there exists a monotone submodular function $v'$ with the following properties:
\begin{itemize}
    \item $v'(S)\le v(S)$ for every $S\subseteq \items$.
    \item For each $j\in[n]$,
    \[
    v'(B_j)=v(B_j)
    \qquad\text{and}\qquad
    \sum_{e\in B_j} v'(\{e\}) = v'(B_j).
    \]
  {(Note: submodularity of $v'$ together with $\sum_{e\in B_j} v'(\{e\}) = v'(B_j)$ implies that $v'$ is additive on every subset of each $B_j$: for every $S\subseteq B_j$, $v'(S)=\sum_{e\in S} v'(\{e\})$.)}
\end{itemize}

\begin{proof}
Fix a bundle index $j$ and fix an \emph{internal order} of the items in $B_j$:
\[
B_j=\{e^j_1,\ldots,e^j_{k_j}\}.
\]
Define the chain prefixes $P^j_r:=\{e^j_1,\ldots,e^j_r\}$ (with $P^j_0=\emptyset$) and define additive weights
\[
w(e^j_r)\ :=\ v(P^j_r)-v(P^j_{r-1}) \ =\ v(e^j_r\mid P^j_{r-1}).
\]
By telescoping,
\[
\sum_{r=1}^{k_j} w(e^j_r) \ =\ v(P^j_{k_j})-v(P^j_0)\ =\ v(B_j).
\]

\medskip
\noindent\textbf{Step 1: $w$ lower-bounds $v$ on subsets of $B_j$.}

{By Claim~\ref{clm:chain-marginal-lower-bound}, applied with $B:=B_j$ under the internal order $e^j_1,\ldots,e^j_{k_j}$, the chain-order marginals $w(e^j_r)=v(e^j_r\mid P^j_{r-1})$ satisfy
\[
\sum_{e\in S} w(e)\ \le\ v(S)\qquad\text{for every }S\subseteq B_j,
\]
with equality on $S=B_j$ .}

\medskip
\noindent\textbf{Step 2: Build $v'$ by iterated truncations.}
Start with $v^{(0)}:=v$.
We will process items in a global sequence that respects the internal order inside each bundle:
process $e^1_1,e^1_2,\ldots,e^1_{k_1}$, then $e^2_1,\ldots,e^2_{k_2}$, and so on.

Suppose we are processing an item $e=e^j_r$.
Let the current function be $v^{(\ell)}$.
We apply Claim~\ref{cl:decreaseSubmodular} to $v^{(\ell)}$ and $e$ with parameter $t:=w(e)$, obtaining $v^{(\ell+1)}$.
By Claim~\ref{cl:decreaseSubmodular}, each step preserves {monotonicity}, submodularity and ensures pointwise domination:
\[
v^{(\ell+1)}(S)\le v^{(\ell)}(S)\le v(S)\qquad\forall S.
\]
Moreover, this truncation forces the \emph{maximum possible marginal} of $e$ into any set to be at most $t=w(e)$,
while leaving all sets not containing $e$ unchanged.

Because the items in different bundles are disjoint, truncating an item in $B_j$ never changes values of sets contained in $B_{j'}$
for $j'\neq j$ (those sets do not contain $e$). Hence we may analyze each bundle separately.

\medskip
\noindent\textbf{Step 3: {The final $v'$ satisfies $v'(B_j)=v(B_j)$ on every bundle $B_j$, and its restriction to each $B_j$ is additive with $v'(\{e\})=w(e)$.}}

Fix $j$ and consider the moment immediately after processing all items in $B_j$ in the iterated truncations.
We claim that for every $r\in\{0,1,\ldots,k_j\}$, the length-$r$ prefix $P^j_r=\{e^j_1,\ldots,e^j_r\}$ of $B_j$ (in the fixed internal order), at every iteration of the process after the truncation of item $e^j_r$, satisfies
\[
v'(P^j_r)\ =\ \sum_{u=1}^{r} w(e^j_u),
\]
and in particular $v'(B_j)=\sum_{u=1}^{k_j} w(e^j_u)=v(B_j)$.

We prove this by induction on $r$.
For $r=0$, $v'(\emptyset)=0$.
Assume it holds for $r-1$.
When processing $e^j_r$, we apply Claim~\ref{cl:decreaseSubmodular} with $t=w(e^j_r)$.
At that moment, the set $P^j_{r-1}$ contains no unprocessed items from $B_j$ and does not contain $e^j_r$. 
Therefore the new value of the prefix including $e^j_r$ is
\[
v'(P^j_r)=v'(P^j_{r-1}\cup\{e^j_r\})=\min\{\,v(\,P^j_r\,),\ v'(P^j_{r-1})+w(e^j_r)\,\}.
\]
By the induction hypothesis, $v'(P^j_{r-1})=\sum_{u=1}^{r-1} w(e^j_u)$, so the second term is
$\sum_{u=1}^{r} w(e^j_u)=v(P^j_r)$ by the definition of $w$ (telescoping along the chain).
Hence the two arguments of the minimum are equal, and we get
\[
v'(P^j_r)=v'(P^j_{r-1})+w(e^j_r)=\sum_{u=1}^{r} w(e^j_u).
\]
{By Claim \ref{cl:decreaseSubmodular}, this value remains the same for all the subsequent truncations.} This completes the induction and proves $v'(B_j)=v(B_j)$.

Finally, because the construction sets $v'(\{e^j_r\})=w(e^j_r)$ at the moment $e^j_r$ is processed, {and this value is preserved by every subsequent truncation: each later step processes some other item $e'\ne e^j_r$, and by Claim~\ref{cl:decreaseSubmodular} such a truncation leaves the value of every set not containing $e'$ unchanged. In particular, of $\{e^j_r\}$, so $v'(\{e^j_r\})=w(e^j_r)$ remains true at the end of the construction.} Therefore, we have
\[
\sum_{e\in B_j} v'(\{e\}) = \sum_{r=1}^{k_j} w(e^j_r) = v'(B_j).
\]
\end{proof}
}

\newcommand{\renderbibliography}{\bibliographystyle{ACM-Reference-Format}\bibliography{references}}

@inproceedings{conf/sigecom/LiptonMMS04,
  author       = {Richard J. Lipton and
                  Evangelos Markakis and
                  Elchanan Mossel and
                  Amin Saberi},
  editor       = {Jack S. Breese and
                  Joan Feigenbaum and
                  Margo I. Seltzer},
  title        = {On approximately fair allocations of indivisible goods},
  booktitle    = {Proceedings 5th {ACM} Conference on Electronic Commerce (EC-2004),
                  New York, NY, USA, May 17-20, 2004},
  pages        = {125--131},
  publisher    = {{ACM}},
  year         = {2004},
  url          = {https://doi.org/10.1145/988772.988792},
  doi          = {10.1145/988772.988792},
  bibsource    = {dblp computer science bibliography, https://dblp.org}
}

@inproceedings{conf/wine/BuLLST24,
  author       = {Xiaolin Bu and
                  Zihao Li and
                  Shengxin Liu and
                  Jiaxin Song and
                  Biaoshuai Tao},
  editor       = {Marios Mavronicolas and
                  Qi Qi and
                  Grant Schoenebeck},
  title        = {Logarithmic Comparison-Based Query Complexity for Fair Division of
                  Indivisible Goods},
  booktitle    = {Web and Internet Economics - 20th International Conference, {WINE}
                  2024, Edinburgh, UK, December 2-5, 2024, Proceedings},
  series       = {Lecture Notes in Computer Science},
  pages        = {348--365},
  publisher    = {Springer},
  year         = {2024},
  url          = {https://doi.org/10.1007/978-3-032-08560-3\_20},
  doi          = {10.1007/978-3-032-08560-3\_20},
  bibsource    = {dblp computer science bibliography, https://dblp.org}
}

@article{journals/talg/AmanatidisMNS17,
  author       = {Georgios Amanatidis and
                  Evangelos Markakis and
                  Afshin Nikzad and
                  Amin Saberi},
  title        = {Approximation Algorithms for Computing Maximin Share Allocations},
  journal      = {{ACM} Trans. Algorithms},
  volume       = {13},
  number       = {4},
  pages        = {52:1--52:28},
  year         = {2017},
  url          = {https://doi.org/10.1145/3147173},
  doi          = {10.1145/3147173},
  bibsource    = {dblp computer science bibliography, https://dblp.org}
}

@article{budish2011combinatorial,
  author    = {Eric Budish},
  title     = {The combinatorial assignment problem: Approximate competitive equilibrium from equal incomes},
  journal   = {Journal of Political Economy},
  volume    = {119},
  number    = {6},
  pages     = {1061--1103},
  year      = {2011},
  publisher = {University of Chicago Press}
}

@article{journals/geb/LehmannLN06,
  author       = {Benny Lehmann and
                  Daniel Lehmann and
                  Noam Nisan},
  title        = {Combinatorial auctions with decreasing marginal utilities},
  journal      = {Games Econ. Behav.},
  volume       = {55},
  number       = {2},
  pages        = {270--296},
  year         = {2006},
  url          = {https://doi.org/10.1016/j.geb.2005.02.006},
  doi          = {10.1016/J.GEB.2005.02.006},
  bibsource    = {dblp computer science bibliography, https://dblp.org}
}

@article{journals/jacm/KurokawaPW18,
  author       = {David Kurokawa and
                  Ariel D. Procaccia and
                  Junxing Wang},
  title        = {Fair Enough: Guaranteeing Approximate Maximin Shares},
  journal      = {J. {ACM}},
  volume       = {65},
  number       = {2},
  pages        = {8:1--8:27},
  year         = {2018},
  url          = {https://doi.org/10.1145/3140756},
  doi          = {10.1145/3140756},
  bibsource    = {dblp computer science bibliography, https://dblp.org}
}

@article{Foley67,
  author  = {Duncan K. Foley},
  title   = {Resource Allocation and the Public Sector},
  journal = {Yale Economic Essays},
  year    = {1967},
  volume  = {7},
  pages   = {45--98}
}

@article{Varian74,
  author  = {Hal R. Varian},
  title   = {Equity, Envy, and Efficiency},
  journal = {Journal of Economic Theory},
  year    = {1974},
  volume  = {9},
  number  = {1},
  pages   = {63--91}
}

@article{journals/corr/abs-2511-13056,
  author       = {Xin Huang and
                  Shengwei Zhou},
  title        = {An {FPTAS} for 7/9-Approximation to Maximin Share Allocations},
  journal      = {CoRR},
  volume       = {abs/2511.13056},
  year         = {2025},
  url          = {https://doi.org/10.48550/arXiv.2511.13056},
  doi          = {10.48550/ARXIV.2511.13056},
  eprinttype   = {arXiv},
  eprint       = {2511.13056},
  bibsource    = {dblp computer science bibliography, https://dblp.org}
}

@article{journals/corr/abs-2505-19961,
  author       = {Uriel Feige},
  title        = {The residual maximin share},
  journal      = {CoRR},
  volume       = {abs/2505.19961},
  year         = {2025},
  url          = {https://doi.org/10.48550/arXiv.2505.19961},
  doi          = {10.48550/ARXIV.2505.19961},
  eprinttype   = {arXiv},
  eprint       = {2505.19961},
  bibsource    = {dblp computer science bibliography, https://dblp.org}
}

@article{journals/corr/abs-2506-21493,
  author       = {Uriel Feige},
  title        = {From multi-allocations to allocations, with subadditive valuations},
  journal      = {CoRR},
  volume       = {abs/2506.21493},
  year         = {2025},
  url          = {https://doi.org/10.48550/arXiv.2506.21493},
  doi          = {10.48550/ARXIV.2506.21493},
  eprinttype   = {arXiv},
  eprint       = {2506.21493},
  bibsource    = {dblp computer science bibliography, https://dblp.org}
}

@inproceedings{conf/aaai/AkramiR25a,
  author       = {Hannaneh Akrami and
                  Nidhi Rathi},
  editor       = {Toby Walsh and
                  Julie Shah and
                  Zico Kolter},
  title        = {Achieving Maximin Share and {EFX/EF1} Guarantees Simultaneously},
  booktitle    = {Thirty-Ninth {AAAI} Conference on Artificial Intelligence, Thirty-Seventh
                  Conference on Innovative Applications of Artificial Intelligence,
                  Fifteenth Symposium on Educational Advances in Artificial Intelligence,
                  {AAAI} 2025, Philadelphia, PA, USA, February 25 - March 4, 2025},
  pages        = {13529--13537},
  publisher    = {{AAAI} Press},
  year         = {2025},
  url          = {https://doi.org/10.1609/aaai.v39i13.33477},
  doi          = {10.1609/AAAI.V39I13.33477},
  bibsource    = {dblp computer science bibliography, https://dblp.org}
}

@inproceedings{conf/soda/BarmanS26,
  author       = {Siddharth Barman and
                  Mashbat Suzuki},
  editor       = {Kasper Green Larsen and
                  Barna Saha},
  title        = {Compatibility of Fairness and Nash Welfare under Subadditive Valuations},
  booktitle    = {Proceedings of the 2026 Annual {ACM-SIAM} Symposium on Discrete Algorithms,
                  {SODA} 2026, Vancouver, BC, Canada, January 11-14, 2026},
  pages        = {1724--1746},
  publisher    = {{SIAM}},
  year         = {2026},
  url          = {https://doi.org/10.1137/1.9781611978971.61},
  doi          = {10.1137/1.9781611978971.61},
  bibsource    = {dblp computer science bibliography, https://dblp.org}
}

@article{journals/corr/abs-2205-05363,
  author       = {Uriel Feige and
                  Alexey Norkin},
  title        = {Improved maximin fair allocation of indivisible items to three agents},
  journal      = {CoRR},
  volume       = {abs/2205.05363},
  year         = {2022},
  url          = {https://doi.org/10.48550/arXiv.2205.05363},
  doi          = {10.48550/ARXIV.2205.05363},
  eprinttype   = {arXiv},
  eprint       = {2205.05363},
  bibsource    = {dblp computer science bibliography, https://dblp.org}
}

@inproceedings{conf/sigecom/AshuriG25,
  author       = {Arash Ashuri and
                  Vasilis Gkatzelis},
  editor       = {Itai Ashlagi and
                  Aaron Roth},
  title        = {Simultaneously Satisfying {MXS} and {EFL}},
  booktitle    = {Proceedings of the 26th {ACM} Conference on Economics and Computation,
                  {EC} 2025, Stanford University, Stanford, CA, USA, July 7-10, 2025},
  pages        = {689--718},
  publisher    = {{ACM}},
  year         = {2025},
  url          = {https://doi.org/10.1145/3736252.3742613},
  doi          = {10.1145/3736252.3742613},
  bibsource    = {dblp computer science bibliography, https://dblp.org}
}

@inproceedings{conf/aaai/BarmanBMN18,
  author       = {Siddharth Barman and
                  Arpita Biswas and
                  Sanath Krishnamurthy and
                  Yadati Narahari},
  editor       = {Sheila A. McIlraith and
                  Kilian Q. Weinberger},
  title        = {Groupwise Maximin Fair Allocation of Indivisible Goods},
  booktitle    = {Proceedings of the Thirty-Second {AAAI} Conference on Artificial Intelligence,
                  (AAAI-18), the 30th innovative Applications of Artificial Intelligence
                  (IAAI-18), and the 8th {AAAI} Symposium on Educational Advances in
                  Artificial Intelligence (EAAI-18), New Orleans, Louisiana, USA, February
                  2-7, 2018},
  pages        = {917--924},
  publisher    = {{AAAI} Press},
  year         = {2018},
  url          = {https://doi.org/10.1609/aaai.v32i1.11463},
  doi          = {10.1609/AAAI.V32I1.11463},
  bibsource    = {dblp computer science bibliography, https://dblp.org}
}

@inproceedings{conf/sigecom/FeigeG25,
  author       = {Uriel Feige and
                  Vadim Grinberg},
  editor       = {Itai Ashlagi and
                  Aaron Roth},
  title        = {Fair allocations with subadditive and {XOS} valuations},
  booktitle    = {Proceedings of the 26th {ACM} Conference on Economics and Computation,
                  {EC} 2025, Stanford University, Stanford, CA, USA, July 7-10, 2025},
  pages        = {160--185},
  publisher    = {{ACM}},
  year         = {2025},
  url          = {https://doi.org/10.1145/3736252.3742514},
  doi          = {10.1145/3736252.3742514},
  bibsource    = {dblp computer science bibliography, https://dblp.org}
}

@inproceedings{conf/ijcai/0001CMS25,
  author       = {George Christodoulou and
                  Vasilis Christoforidis and
                  Symeon Mastrakoulis and
                  Alkmini Sgouritsa},
  title        = {Maximin Share Guarantees for Few Agents with Subadditive Valuations},
  booktitle    = {Proceedings of the Thirty-Fourth International Joint Conference on
                  Artificial Intelligence, {IJCAI} 2025, Montreal, Canada, August 16-22,
                  2025},
  pages        = {3788--3795},
  publisher    = {ijcai.org},
  year         = {2025},
  url          = {https://doi.org/10.24963/ijcai.2025/421},
  doi          = {10.24963/IJCAI.2025/421},
  bibsource    = {dblp computer science bibliography, https://dblp.org}
}

@article{journals/corr/abs-2303-12444,
  author       = {Gilad Ben Uziahu and
                  Uriel Feige},
  title        = {On Fair Allocation of Indivisible Goods to Submodular Agents},
  journal      = {CoRR},
  volume       = {abs/2303.12444},
  year         = {2023},
  url          = {https://doi.org/10.48550/arXiv.2303.12444},
  doi          = {10.48550/ARXIV.2303.12444},
  eprinttype   = {arXiv},
  eprint       = {2303.12444},
  bibsource    = {dblp computer science bibliography, https://dblp.org}
}

@inproceedings{conf/aaai/AkramiR25,
  author       = {Hannaneh Akrami and
                  Nidhi Rathi},
  editor       = {Toby Walsh and
                  Julie Shah and
                  Zico Kolter},
  title        = {Epistemic {EFX} Allocations Exist for Monotone Valuations},
  booktitle    = {Thirty-Ninth {AAAI} Conference on Artificial Intelligence, Thirty-Seventh
                  Conference on Innovative Applications of Artificial Intelligence,
                  Fifteenth Symposium on Educational Advances in Artificial Intelligence,
                  {AAAI} 2025, Philadelphia, PA, USA, February 25 - March 4, 2025},
  pages        = {13520--13528},
  publisher    = {{AAAI} Press},
  year         = {2025},
  url          = {https://doi.org/10.1609/aaai.v39i13.33476},
  doi          = {10.1609/AAAI.V39I13.33476},
  bibsource    = {dblp computer science bibliography, https://dblp.org}
}

@article{journals/siamdm/PlautR20,
  author       = {Benjamin Plaut and
                  Tim Roughgarden},
  title        = {Almost Envy-Freeness with General Valuations},
  journal      = {{SIAM} J. Discret. Math.},
  volume       = {34},
  number       = {2},
  pages        = {1039--1068},
  year         = {2020},
  url          = {https://doi.org/10.1137/19M124397X},
  doi          = {10.1137/19M124397X},
  bibsource    = {dblp computer science bibliography, https://dblp.org}
}

@inproceedings{conf/wine/FeigeST21,
  author       = {Uriel Feige and
                  Ariel Sapir and
                  Laliv Tauber},
  editor       = {Michal Feldman and
                  Hu Fu and
                  Inbal Talgam{-}Cohen},
  title        = {A Tight Negative Example for {MMS} Fair Allocations},
  booktitle    = {Web and Internet Economics - 17th International Conference, {WINE}
                  2021, Potsdam, Germany, December 14-17, 2021, Proceedings},
  series       = {Lecture Notes in Computer Science},
  pages        = {355--372},
  publisher    = {Springer},
  year         = {2021},
  url          = {https://doi.org/10.1007/978-3-030-94676-0\_20},
  doi          = {10.1007/978-3-030-94676-0\_20},
  bibsource    = {dblp computer science bibliography, https://dblp.org}
}

@article{journals/teco/CaragiannisKMPS19,
  author       = {Ioannis Caragiannis and
                  David Kurokawa and
                  Herv{\'{e}} Moulin and
                  Ariel D. Procaccia and
                  Nisarg Shah and
                  Junxing Wang},
  title        = {The Unreasonable Fairness of Maximum Nash Welfare},
  journal      = {{ACM} Trans. Economics and Comput.},
  volume       = {7},
  number       = {3},
  pages        = {12:1--12:32},
  year         = {2019},
  url          = {https://doi.org/10.1145/3355902},
  doi          = {10.1145/3355902},
  bibsource    = {dblp computer science bibliography, https://dblp.org}
}

@article{journals/jacm/ChaudhuryGM24,
  author       = {Bhaskar Ray Chaudhury and
                  Jugal Garg and
                  Kurt Mehlhorn},
  title        = {{EFX} Exists for Three Agents},
  journal      = {J. {ACM}},
  volume       = {71},
  number       = {1},
  pages        = {4:1--4:27},
  year         = {2024},
  url          = {https://doi.org/10.1145/3616009},
  doi          = {10.1145/3616009},
  bibsource    = {dblp computer science bibliography, https://dblp.org}
}

@article{journals/ai/GhodsiHSSY22,
  author       = {Mohammad Ghodsi and
                  Mohammad Taghi Hajiaghayi and
                  Masoud Seddighin and
                  Saeed Seddighin and
                  Hadi Yami},
  title        = {Fair allocation of indivisible goods: Beyond additive valuations},
  journal      = {Artif. Intell.},
  volume       = {303},
  pages        = {103633},
  year         = {2022},
  url          = {https://doi.org/10.1016/j.artint.2021.103633},
  doi          = {10.1016/J.ARTINT.2021.103633},
  bibsource    = {dblp computer science bibliography, https://dblp.org}
}

@inproceedings{conf/ijcai/AkramiGST23,
  author       = {Hannaneh Akrami and
                  Jugal Garg and
                  Eklavya Sharma and
                  Setareh Taki},
  title        = {Simplification and Improvement of {MMS} Approximation},
  booktitle    = {Proceedings of the Thirty-Second International Joint Conference on
                  Artificial Intelligence, {IJCAI} 2023, 19th-25th August 2023, Macao,
                  SAR, China},
  pages        = {2485--2493},
  publisher    = {ijcai.org},
  year         = {2023},
  url          = {https://doi.org/10.24963/ijcai.2023/276},
  doi          = {10.24963/IJCAI.2023/276},
  bibsource    = {dblp computer science bibliography, https://dblp.org}
}

@inproceedings{conf/soda/AkramiG24,
  author       = {Hannaneh Akrami and
                  Jugal Garg},
  editor       = {David P. Woodruff},
  title        = {Breaking the 3/4 Barrier for Approximate Maximin Share},
  booktitle    = {Proceedings of the 2024 {ACM-SIAM} Symposium on Discrete Algorithms,
                  {SODA} 2024, Alexandria, VA, USA, January 7-10, 2024},
  pages        = {74--91},
  publisher    = {{SIAM}},
  year         = {2024},
  url          = {https://doi.org/10.1137/1.9781611977912.4},
  doi          = {10.1137/1.9781611977912.4},
  bibsource    = {dblp computer science bibliography, https://dblp.org}
}

@article{journals/teco/BarmanK20,
  author       = {Siddharth Barman and
                  Sanath Kumar Krishnamurthy},
  title        = {Approximation Algorithms for Maximin Fair Division},
  journal      = {{ACM} Trans. Economics and Comput.},
  volume       = {8},
  number       = {1},
  pages        = {5:1--5:28},
  year         = {2020},
  url          = {https://doi.org/10.1145/3381525},
  doi          = {10.1145/3381525},
  bibsource    = {dblp computer science bibliography, https://dblp.org}
}

@inproceedings{conf/sigecom/GhodsiHSSY18,
  author       = {Mohammad Ghodsi and
                  Mohammad Taghi Hajiaghayi and
                  Masoud Seddighin and
                  Saeed Seddighin and
                  Hadi Yami},
  editor       = {{\'{E}}va Tardos and
                  Edith Elkind and
                  Rakesh Vohra},
  title        = {Fair Allocation of Indivisible Goods: Improvements and Generalizations},
  booktitle    = {Proceedings of the 2018 {ACM} Conference on Economics and Computation,
                  Ithaca, NY, USA, June 18-22, 2018},
  pages        = {539--556},
  publisher    = {{ACM}},
  year         = {2018},
  url          = {https://doi.org/10.1145/3219166.3219238},
  doi          = {10.1145/3219166.3219238},
  bibsource    = {dblp computer science bibliography, https://dblp.org}
}

@inproceedings{conf/soda/HeidariKSS26,
  author       = {Ehsan Heidari and
                  Alireza Kaviani and
                  Masoud Seddighin and
                  AmirMohammad Shahrezaei},
  editor       = {Kasper Green Larsen and
                  Barna Saha},
  title        = {Improved Maximin Share Guarantee for Additive Valuations},
  booktitle    = {Proceedings of the 2026 Annual {ACM-SIAM} Symposium on Discrete Algorithms,
                  {SODA} 2026, Vancouver, BC, Canada, January 11-14, 2026},
  pages        = {2239--2290},
  publisher    = {{SIAM}},
  year         = {2026},
  url          = {https://doi.org/10.1137/1.9781611978971.81},
  doi          = {10.1137/1.9781611978971.81},
  bibsource    = {dblp computer science bibliography, https://dblp.org}
}

@article{journals/tcs/AmanatidisMN20,
  author       = {Georgios Amanatidis and
                  Evangelos Markakis and
                  Apostolos Ntokos},
  title        = {Multiple birds with one stone: Beating 1/2 for {EFX} and {GMMS} via
                  envy cycle elimination},
  journal      = {Theor. Comput. Sci.},
  volume       = {841},
  pages        = {94--109},
  year         = {2020},
  url          = {https://doi.org/10.1016/j.tcs.2020.07.006},
  doi          = {10.1016/J.TCS.2020.07.006},
  bibsource    = {dblp computer science bibliography, https://dblp.org}
}

@article{fan2015exponential,
  author       = {Xiequan Fan and
                  Ion Grama and
                  Quansheng Liu},
  title        = {Exponential inequalities for martingales with applications},
  journal      = {Electron. J. Probab.},
  volume       = {20},
  pages        = {1--22},
  year         = {2015},
  doi          = {10.1214/EJP.v20-3496}
}

@article{freedman1975tail,
  author       = {David A. Freedman},
  title        = {On tail probabilities for martingales},
  journal      = {Ann. Probab.},
  volume       = {3},
  number       = {1},
  pages        = {100--118},
  year         = {1975},
  doi          = {10.1214/aop/1176996452}
}

@article{journals/tcs/GourvesM19,
  author       = {Laurent Gourv{\`{e}}s and
                  J{\'{e}}r{\^{o}}me Monnot},
  title        = {On maximin share allocations in matroids},
  journal      = {Theor. Comput. Sci.},
  volume       = {754},
  pages        = {50--64},
  year         = {2019},
  url          = {https://doi.org/10.1016/j.tcs.2018.05.018},
  doi          = {10.1016/J.TCS.2018.05.018},
  bibsource    = {dblp computer science bibliography, https://dblp.org}
}

@article{journals/corr/abs-2604-18216,
  author       = {Hannaneh Akrami and
                  Alexander Mayorov and
                  Kurt Mehlhorn and
                  Shreyas Srinivas and
                  Christoph Weidenbach},
  title        = {A Counterexample to {EFX} $n \ge 3$ Agents, $m \ge n + 5$ Items, Submodular Valuations via {SAT}-Solving},
  journal      = {CoRR},
  volume       = {abs/2604.18216},
  year         = {2026},
  url          = {https://doi.org/10.48550/arXiv.2604.18216},
  doi          = {10.48550/ARXIV.2604.18216},
  eprinttype   = {arXiv},
  eprint       = {2604.18216},
  bibsource    = {dblp computer science bibliography, https://dblp.org}
}

@article{mackenzie2026counterexamplesefxsubmodularsubadditive,
  author       = {Simon Mackenzie and
                  Mashbat Suzuki},
  title        = {Counterexamples to {EFX} for Submodular and Subadditive Valuations},
  journal      = {CoRR},
  volume       = {abs/2605.06451},
  year         = {2026}
}

\renderbackmatter

\end{document}